\documentclass[journal]{IEEEtran}
\IEEEoverridecommandlockouts    
\usepackage[noadjust]{cite}
\usepackage{hyperref}
\usepackage{soul}
\usepackage{rotating}
\usepackage{chngcntr}
\usepackage{apptools}
\AtAppendix{\counterwithin{thm}{section}}
\usepackage{graphicx}
\usepackage{graphics}
\usepackage{amsthm}
\usepackage{amssymb}
\usepackage{amsmath}
\usepackage{color}
\usepackage{xspace}
\usepackage{algpseudocode}
\usepackage{bbm}
\usepackage{comment}
\usepackage{algorithmicx}
\usepackage{subfig}
\usepackage{psfrag}
\usepackage{tikz}
\usetikzlibrary{shapes,arrows}
\usetikzlibrary{positioning}
    \usepackage{graphicx}
    \usepackage{dcolumn}
    \usepackage{bm}
    \usepackage{amsmath} 

\tikzstyle{block}=[draw opacity=0.7,line width=1.4cm]
\DeclareMathAlphabet{\mathpzc}{OT1}{pzc}{m}{it}
\definecolor{CranJ}{cmyk}{0,0.69,0.54,0.04} %cranberry jello
\definecolor{PinkJ}{cmyk}{0,0.71,0.43,0.12} %pink jeep
\definecolor{Cran}{cmyk}{0,0.73,0.41,0.29} %cranberry 
\definecolor{VRed}{cmyk}{0,0.75,0.25,0.2} %violetred
\definecolor{ORed}{cmyk}{0,0.75,0.75,0} %orangered4
\definecolor{CBlue}{cmyk}{1,0.25,0,0} %curacao	
\title{ \LARGE \bf  A Study  on  Accelerating  Average Consensus Algorithms\\ Using Delayed Feedback}

\author{Hossein Moradian, \emph{Student Member, IEEE} and Solmaz S. Kia, \emph{Senior Member, IEEE}
  \thanks{The authors are with the Department of Mechanical and Aerospace Engineering, University of California Irvine, Irvine, CA 92697,  
    {\tt\small \{hmoradia,solmaz\}@uci.edu}}%
}

\newcommand{\ee}{\operatorname{e}}
\newcommand{\VV}{\mathcal{V}}
\newcommand{\EE}{\mathcal{E}}
\newcommand{\GG}{\mathcal{G}}

\newcommand{\real}{{\mathbb{R}}}
\newcommand{\reals}{{\mathbb{R}}}
\newcommand{\integer}{{\mathbb{Z}}}

\newcommand{\complex}{{\mathbb{C}}} 
\newcommand{\realpositive}{{\mathbb{R}}_{>0}}
\newcommand{\realnonnegative}{{\mathbb{R}}_{\ge 0}}

\newcommand{\eps}{\epsilon}

\newcommand{\re}[1]{\operatorname{Re}(#1)}
\newcommand{\im}[1]{\operatorname{Im}(#1)}

\newcommand{\until}[1]{\in\{1,\dots,#1\}}

\newcommand{\argmax}{\operatorname{argmax}}
\newcommand{\vect}[1]{\boldsymbol{\mathbf{#1}}}
\newcommand{\vectsf}[1]{\vect{\mathsf{#1}}}
\newcommand{\Bvect}[1]{\bar{\boldsymbol{\mathbf{#1}}}}

\newcommand{\dvect}[1]{\dot{\vect{#1}}}
\newcommand{\dvectsf}[1]{\dot{\vectsf{#1}}}

\newcommand{\Diag}[1]{\operatorname{Diag}(#1)}

 \newcommand{\boxend}{\hfill \ensuremath{\Box}}

\newcommand{\margin}[1]{\marginpar{\color{red}\tiny\ttfamily#1}}

\newcommand{\solmaz}[1]{{\color{red}#1}}
\newcommand{\hossein}[1]{{\color{blue}#1}}

\newtheorem{thm}{Theorem}[section]

\newtheorem{rem}{Remark}[section]

\newtheorem{lem}{Lemma}[section]

\newcommand{\oprocendsymbol}{\hbox{$\bullet$}}
\newcommand{\oprocend}{\relax\ifmmode\else\unskip\hfill\fi\oprocendsymbol}


\begin{document}%\fontsize{10}{12.5}\rm
\maketitle
\pagenumbering{roman}
\begin{abstract}
In this paper, we study accelerating a  Laplacian-based dynamic average consensus algorithm by splitting  the conventional delay-free disagreement feedback into weighted summation of a current and an outdated term.
We determine for what weighted sum there exists a range of time delay that results in the higher rate of convergence for the algorithm. For such weights, using the Lambert W function, we obtain the rate increasing range of the time delay, the maximum reachable rate and comment on the value of the corresponding maximizer delay. We also study the effect of use of outdated feedback on the control effort of the agents and show that only for some specific affine combination of the immediate and outdated feedback the control effort of the agents does not go beyond that of the delay-free algorithm. Additionally, we demonstrate that using outdated feedback does not increase the steady state tracking error of the average consensus algorithm.
Lastly, we determine the optimum combination of the current and the outdated feedback weights to achieve maximum increase in the rate of convergence without increasing the control effort of the agents. We demonstrate our results through a numerical example.

\end{abstract}

%\input{Automatica-header.tex}
%\AtBeginDocument{%
%   \setlength\abovedisplayskip{5pt}
%   \setlength\belowdisplayskip{5pt}}
%  \setlength{\parskip}{1.5mm}
%\input{Automatica-abstract.tex}
%\renewcommand\baselinestretch{1}

\vspace{-0.1in}
\section{Introduction}
\vspace{-0.05in}
The average consensus problem for a group of networked agents each endowed with a  reference input %$\mathsf{r}^i$
signal (dynamic or static) is defined as designing a distributed interaction policy for each agent %$i\until{N}$ 
such that a local agreement state    converges asymptotically to the average of the reference signals across the network.
%$\mathsf{x}^{\text{avg}}=
% $\frac{1}{N}\sum_{j=1}^N \mathsf{r}^j$.  %\solmaz{add one or two examples from dynamic consensus case}
%Average consensus over multi-agent systems has been studied extensively in the literature (see e.g.,~\cite{ROS-RMM:04,WR-RWB:05,LX-SB:04},~\cite{ROS-JAF-RMM:07}).
For this problem, in continuous time domain, when the reference signals of all the agents are static, the well-known distributed solution is the  Laplacian consensus algorithm~\cite{ROS-RMM:04,WR-RWB:05,LX-SB:04,ROS-JAF-RMM:07}. In the Laplacian consensus, each agent initializes its first order integrator dynamics with its local reference value and uses the weighted sum of the  difference between its local state and those of its neighbors (disagreement feedback) to drive its local dynamics to the average of the reference signals across the network. When the reference signals are dynamics, agents use a combination of the Laplacian input and their local reference signal and/or its derivative to drive their local integrator dynamics; see~\cite{SSK-BVS-JC-RAF-KML-SM} for examples of dynamic average consensus algorithms. Average consensus algorithms are of interest in various multi-agent applications such as sensor fusion~\cite{ROS-JSS:05,ROS:07,ATK-JAF-AKR:13,WR-UMA:17}, robot coordination~\cite{PY-RAF-KML:08,YC-SSK:19b}, formation control~\cite{JAF-RMM:04}, distributed optimal resource allocation~\cite{AC-JC:14-auto,SSK:17}, distributed estimation~\cite{SM:07} and distributed tracking~\cite{PY-RAF-KML:07}. For these cooperative tasks, it is highly desired that the consensus among the agents is obtained fast, i.e., the consensus  algorithm converges fast. For a connected network  with undirected communication, it is well understood that~the  convergence rate of the average consensus algorithms is associated with the connectivity of the graph~\cite{MF:73}, specified by the smallest non-zero eigenvalue of the Laplacian matrix~\cite{ROS-RMM:04,SSK-BVS-JC-RAF-KML-SM}.

Given this connection, 
various efforts such as optimal adjacency weight selection for a given topology by maximizing the smallest non-zero eigenvalue of the Laplacian matrix~\cite{LX-SB:04,SB-AG-BP-DS:06} or rewiring the graph to create topologies such as small-world network~\cite{SK-JMF:06,PH-JSB-VG:08} with high connectivity have been proposed in the literature. 
%{\color{red}Add the references here
%~\url{https://users.ece.cmu.edu/~soummyak/Asilomar_06.pdf}\\
%~~~\url{https://ece.umd.edu/~baras/publications/papers/2008/2008_Hovareshti_AverageConsensus.pdf}
%}
%Due to its importance as a major criterion of performance, many efforts are done to increase the rate of convergence of average consensus algorithm, mostly by executing matrix optimization algorithms. These approaches aim to  find the optimal weight configuration of the communications over the network by maximizing second eigenvalue  of the Laplacian matrix~\cite{LX-SB:04,SB-AG-BP-DS:06}. 
%\margin{\solmaZ{I don't exactly follow what you mean here. Are they higher order dynamics?} Using state prediction approach~\cite{TCA-BNO-MJC:09}
%\hossein{The authors in~\cite{TCA-BNO-MJC:09} have examined accelerating the convergence rate
%of a distributed average consensus algorithms by changing the
%state update to a  combination of the standard consensus
%iteration and a linear prediction. The algorithm is second order in the discrete-time form. For the predictor they used linear combination of previous iterations.}
%} %Although practical, most of these methods require extra.
In this paper, we study use of outdated disagreement feedback to increase the convergence rate of a dynamic average consensus algorithm. Our method can be applied in conjunction with the aforementioned topology designs to maximize the acceleration effect. 
%We show that this approach does not expose any extra control effort to the agent to achieve consensus.% with static reference signals. Moreover, for the dynamic consensus with time-varying reference signals the steady state error does not increase by using this approach. 

\begin{comment}
\hossein{Wei Ren:
\begin{itemize}
    \item Considering the special case $\mathsf{k=0.5}$. 
    \item showing that there exist a range with higher rate of convergence
    \item showing that there is no extra control effort for delayed feedback.
\end{itemize}
Riphat Sipahi:
\begin{itemize}
    \item Considering the special case $\mathsf{k=1}$.
    \item finding the maximum rate of convergence for scalar case(For matrix case there is no closed form solution)
    \item showing the trend of rate of convergence(increasing and decreasing form)
\end{itemize}
Our TAC paper:
\begin{itemize}
    \item Considering the special case $\mathsf{k=1}$.
    \item showing the trend of rate(increasing and decreasing form)
    \item finding the maximum rate as a closed form for matrix with real eigenvalues and specifying a region for complex eigenvalues. 
\end{itemize}

}
\end{comment}
In performance analysis of dynamical systems intuition inadvertently ties time delay to sluggishness and adverse effects on system response. However, some work such as~\cite{BG-SM-MHS:98,YG-MS-CS-NM:16,MC-DAS-EMY:06,ZM-YC-WR:10,AGU:17,YC-WR:10,HM-SSK:18b, WQ-RS:13} point to the positive effect of time delay on increasing stability margin and rate of convergence of time-delayed systems. Specifically, the positive effect of time-delayed feedback in accelerating the  convergence of the static average consensus Laplacian algorithm is reported in~\cite{HM-SSK:18b,WQ-RS:13,YC-WR:10}. 
The study in~\cite{HM-SSK:18b} and~\cite{WQ-RS:13} consider delaying  the immediate Laplacian disagreement feedback,  and show that when the network topology is connected there always exists a range of delay $(0,\tilde{\tau})$ such that the rate of convergence of the modified algorithm is faster. The technical results also include specifying $\tilde{\tau}$ and also showing that the maximum attainable convergence rate due to employing delayed feedback is the Euler number times the rate without delay.~\cite{HM-SSK:18b} also specifies the delay for which maximum convergence rate is attained. However, the effect of use of outdated feedback on the control effort of the agents is left unexplored in~\cite{HM-SSK:18b} and~\cite{WQ-RS:13}. This study is of importance because one can always argue that the convergence rate of the original Laplacian algorithm can be increased by multiplying the Laplacian input with a gain greater than one. But, this choice leads to  increase in the control effort of the agents. On the other hand,~\cite{YC-WR:10} studies a modified Laplacian algorithm where the immediate Laplacian disagreement feedback input is broken in half and one half is replaced by  outdated delayed feedback. For this modified algorithm,~\cite{YC-WR:10} shows that it is possible to increase the convergence rate for some values of delay. They also show that this increase in rate is without increasing the control effort of the agents. However,~\cite{YC-WR:10} falls short of specifying the exact range of delay for which the convergence rate can be increased by employing the outdated feedback and also quantifying the maximum rate and its corresponding maximizer delay.

In this paper, we study the use of outdated disagreement feedback to increase the convergence rate of the dynamic average consensus algorithm of~\cite{dps-ros-rmm:05b}. This algorithm, when the reference signal of the agents are all static, simplifies to the static Laplacian average consensus algorithm~\cite{ROS-JAF-RMM:07}. In our study, we split the disagreement Laplacian feedback into two components of immediate and outdated feedback. However instead of equal contribution, we consider the affine combination of the current and outdated feedback to investigate the effect of the relative size of the outdated and immediate feedback terms on the induced acceleration. Our comprehensive study  includes~\cite{YC-WR:10},~\cite{WQ-RS:13} and~\cite{HM-SSK:18b} as special cases. We note here that the analysis methods used in~\cite{YC-WR:10},~\cite{WQ-RS:13} and~\cite{HM-SSK:18b} do not generalize to study the case of affine combination of the immediate and outdated feedback. This is due to the technical challenges involved with study of the variation of the infinite number of the roots of the characteristic equation of the linear time delayed systems with delay, which often are resolved via methods that conform  closely to the specific algebraic structure of the system  under study. We recall here that the exact value of the worst convergence rate of a linear time-invariant system,  with or without delay, is determined by the magnitude of the real part of the right most root of its characteristic equation~\cite{TH-ZL-YS:03,SD-JN-AGU:11}.

We start our study by characterizing the admissible range of delay for which the average consensus tracking is maintained. Then, we show that for the delays in the admissible range, the ultimate tracking error of our modified average consensus algorithm of interest is not affected by use of outdated feedback regardless of the affine combination's split factor. However, we show that the control effort of the agents does not increase only for a specific range of the split factor of the affine combination of the outdated and immediate feedback. Our results also specify (a) for what values of the system parameters the rate of convergence in the presence of delay can increase, (b) the exact values of delay for which the rate of convergence increases, and (c) the optimum value of $\tau$ corresponding to the maximum rate of convergence in the presence of delay. In light of all the aforementioned study, we summarize our results in the remark that discuses the trade off between the performance (maximum convergence rate) and the robustness of the algorithm to the delay as well as the level of control effort. Our study relies on use of the Lambert W function~\cite{HS-TM:06,RMC-GHG-DEGH-DJJ-DEK:96} to obtain the exact value of the characteristic roots of the internal dynamics of our dynamic consensus algorithm. Via careful study of variation of the right most root in the complex plan with respect to delay we then proceed to conduct our study to establish our~results.

\emph{Organization}:
%The rest of this paper is organized as follows. 
Notations and preliminaries including a brief review of the relevant properties of the Lambert W function and the graph theoretic definitions are given in Section~\ref{sec::prelim}.  Problem definitions and the objective statements are given in Section~\ref{sec::Prob_formu}, while the main results are given in~Section~\ref{sec::main}.% presents our main results on accelerating the average consensus algorithm.
Numerical simulations to illustrate our results are given in Section~\ref{sec::Num_ex}.   Section~\ref{sec::conclu} summarizes our concluding remarks. Finally, the appendices contain the auxiliary lemmas that we use to develop our main~results.

\section{Notations and Preliminaries}\label{sec::prelim}
In this section, we review our notations, definitions and auxiliary results that we use in our developments.

We let $\reals$, $\realpositive$, $\realnonnegative$, $\mathbb{Z}$, and $\mathbb{C}$
denote the set of real, positive real, non-negative real, integer, and complex numbers, respectively. Given $i,j\in\mathbb{Z}$ with $i<j$, we define $\mathbb{Z}_i^j=\{i,i+1,
\cdots,j\}$. For  $s\in\mathbb{C}$, $\re{s}$ and $\im{s}$ represent,  respectively, the real and imaginary parts of $s$. Moreover, $|s|=\sqrt{\re{s}^2+\im{s}^2}$ and $\arg(s)=\text{atan2}(\im{s},\re{s})$. For any vector $\vect{x}\in\real^n$, we let $\|\vect{x}\|=\sqrt{x_1^2+\cdots+x_n^2}$ and $\|\vect{x}\|_\infty=\max\{x_i\}_{i=1}^n$.
For a measurable locally essentially bounded function $\vect{u}:\real_{\geq0}\to\real^m$, we define
$|\vect{u}|_{\infty}= \text{ess}\sup\{\|\vect{u}(t)\|_\infty,~t\geq0\}$. For a  matrix $\vect{A}$, its $i^{th}$ row is denoted by $[\vect{A}]_i$. 
%\solmaz{we need to define $\|u\|$ in $\text{ess}\sup\{\|\vect{u}(t)\|,~t\geq0\}$. what we need is $\|u\|_{infty}$}
% represent, respectively, the magnitude, i.e., $|s|=\sqrt{\re{s}^2+\im{s}^2}$ and the argument, i..e, $\arg(s)=\text{atan2}(\im{s},\re{s})$ of $s\in\mathbb{C}$.

For a linear time-delayed system, \emph{admissible delay range} is the range of time delay for which the internal dynamics of the system is %exponentially 
stable. We recall that for linear time-delayed systems with exponentially stable dynamics when delay is set to zero, by virtue of the \emph{continuity stability property} theorem~\cite[Proposition 3.1]{SN:01}, the admissible delay range is a connected range $(0,\bar{\tau})\subset\real_{>0}$ where $\bar{\tau}\in\real_{>0}$ is the \emph{critical} delay bound beyond which the system is always~unstable.   
\begin{lem}[Admissible delay bound for a scalar time-delayed system {\cite[Proposition~3.15]{SN:01}}]
%\margin{this was a Theorem initially, check wherever it is called to make sure it is referenced as Lemma not theorem.}
 \label{thm::stability_con}
 {\rm Consider 
 \begin{align}\label{DDE_scalar}
\dot{x}(t)&=\mathsf{a}\,x(t-\tau)+\mathsf{b}\,x(t),\quad \quad t\in\real_{\geq0},\\
x(\eta) &\in\real, \quad\quad  \quad\quad\quad  \quad \quad\quad  \quad \eta\in [-\tau,0],\nonumber
\end{align}
where $\mathsf{a}\in\real\backslash\{0\}$ and $\mathsf{a}+\mathsf{b}<0$. Then, the following assertions hold %for system~\eqref{DDE_sys_scalar} with $\alpha\in\real\backslash\{0\}$ and $\alpha+\mathsf{b}<0$:
\begin{itemize}
    \item[(a)] For  $\mathsf{b}\leq-|\mathsf{a}|$, system~\eqref{DDE_scalar} is exponentially stable independent of the value of $\tau\in\real_{\geq0}$, i.e., $\bar{\tau}=\infty$.
    \item[(b)] For  $\mathsf{a}<-|\mathsf{b}|$, system~\eqref{DDE_scalar} is exponentially stable if and only if $\tau\in[0,\bar{\tau})$ where 
   \rm{ \begin{align}\label{eq::Admis_range_tau}
   \bar{\tau}=\frac{\text{arccos}(-\frac{\mathsf{b}}{\mathsf{a}})}{\sqrt{\mathsf{a}^2-\mathsf{b}^2}}.
    \end{align}}
\end{itemize}}
\boxend
\end{lem}

Lambert $W$ function specifies the solutions of $s\,\ee^{s}=z$ for a given $z\in\complex$, i.e., $s=W(z)$. It is a multivalued function with  infinite number of solutions denoted by $W_k(z)$, $k\in\integer$, where $W_k$ is called the $k^{\text{th}}$ branch of $W$ function.  For any $z\in\complex$, $W_k(z)$
can readily be evaluated in Matlab or Mathematica. 
Below are some of the intrinsic properties of the Lambert $W$ function, which we use (see~\cite{HS-TM:06,RMC-GHG-DEGH-DJJ-DEK:96}),
\begin{subequations}
\begin{align}
 \lim_{z\to 0}W_{k}(z)/z&=1, \label{eq::limWz-z}\\
\text{d}\,W_{k}(z)/\text{d}\,z&=1/(z+\ee^{W_{k}(z)}),\quad\quad\text{for~}z\neq 1/\ee,\label{eq::lambert_derivative}
\end{align}
\end{subequations}
for $k\in\mathbb{Z}$. For any $z\!\in\!\real$, the value of all the branches of the Lambert $W$ function except for the branch $0$ and the branch $-1$ are complex (non-zero imaginary part). Moreover, the zero branch satisfies $W_0(-1/\ee)\!=\!-1$, $W_0(0)\!=\!0$ 
~and  %(see Fig.~  \ref{fig::real_lam}),
\begin{subequations}\label{eq::W_0_facts}
\begin{align}
W_0(z)&\in\real,\,~\quad\quad\quad\quad\quad z\in[-1/\ee,\infty),  \label{eq::W_0_facts_a} \\
\im{W_0(z)}&\in(-\pi,\pi)\backslash\{0\},\quad z\in\complex\backslash[-1/\ee,\infty),\label{eq::W_0_facts_b}\\
\re{W_0(z)}&>-1, \quad\quad\quad\quad\quad z\in\real\backslash\{-1/\ee\},
\label{eq::W_0_facts_c}
%\hossein{\re{W_k(z)}&<-1,\quad for \quad k\in\integer/\{0\}\quad z\in(-\frac{1}{\ee},0)
%\label{eq::W_0_facts_d}\\
%\re{W_0(z)}&=\re{W_{-1}(z)} \quad\quad\re{W_k(z)}<-1,\nonumber\\&\quad for \quad k\in\integer/\{0,-1\}\quad z\in(-\frac{\pi}{2},-\frac{1}{\ee})\label{eq::W_0_facts_e}
%}
\end{align}
\end{subequations}

\begin{lem}[Maximum real part of Lambert W function~\cite{HS-TM:06}]\label{lem::LambertReal}
{\rm For any $z\in\complex$, the following holds 
%in $s\,\ee^{s}=z$, the solutions $s=W(z)$ satisfy
\begin{align}\label{eq::max-W0}
\re{W_0(z)}\geq \max\big\{\re{W_k(z)}|~k\in\integer\backslash\{0\}\big\}.
\end{align}
The equality holds between branch $0$ and $-1$ over  $z\in(-\infty,-\frac{1}{\ee})$  where we have $\re{W_0(z)}=\re{W_{-1}(z)}$.
\boxend}
\end{lem}
%Figure~  \ref{fig::real_lam} shows an example highlighting the result of Lemma~  \ref{lem::LambertReal}. 
%Other properties of the Lambert $W$ function which we invoke are listed next. 
%We close this section by the auxiliary result below.
\begin{lem}[$W_0(\mathsf{x})$ is an increasing function of $x\in\real_{>0}$]\label{lem::W0-increase-xpos} 
{\rm For any $\mathsf{x},\mathsf{y}\in\real_{>0}$ if $\mathsf{x}< \mathsf{y}$, then $W_0(\mathsf{x})<W_0(\mathsf{y})$.}
\end{lem}
%\vspace{-0.2in}
\begin{proof}
The proof follows from the fact that for $\mathsf{x}\in\real_{>0}$, $W_0(\mathsf{x})\in\real_{>0}$.  Therefore, 
$\frac{\text{d}\,W_0(\mathsf{x})}{\text{d}\,\mathsf{x}}=\frac{1}{\mathsf{x}+\ee^{W_0(\mathsf{x})}}>0$.
\end{proof}

\begin{comment}

\begin{figure}
    \centering
  % \includegraphics[trim={0pt 0 0 0},clip,scale=0.6]{Fig/lambert}
    \includegraphics[trim={28pt 0 0 0},scale=0.35]{acc1.png}
\caption{ The solid thick line shows $0$ branch while the thin dashed lines show the other branches in $k\!=\!\{-5,-4,\cdots,4,5\}$ (some of the branches overlap)%\NeedATT{I think we can exclude this figure to fit 6-page}. 
}
%   \solmaz{your figure is also sticking out of the margin. Also you need to use a consistent n notation, in text and figure, we are using $W(z)$ in text  and in the figure you use $W(x)$. You need to change your figure. Your need to also correct your y-axis label. }
    \label{fig::real_lam}
\end{figure}
\end{comment}

\begin{comment}
\begin{figure}
    \centering
   \includegraphics[trim=35 5 10 25,clip,scale=0.26]{branches.png}
    %  \includegraphics[scale=0.7]{Fig/branch_1.png}
   %\includegraphics[trim={0pt 0 0 0},clip,scale=0.6]{Fig/lambert8}
\caption{Range of various branches of Lambert function\cite{CH-YCC:05}.
}   \label{fig::branches_lam}
\end{figure}
\end{comment}

We follow~\cite{FB-JC-SM:09} to define our graph related terminologies and notations. In a network of $N$ agents, we model the inter-agent interaction topology by the undirected connected graph
 $\mathcal{G}(\mathcal{V},\mathcal{E},\vectsf{A})$ where $\mathcal{V}$ is the node set, $\mathcal{E}\subset\mathcal{V}\times \mathcal{V}$ is the edge set and $\vectsf{A}=[a_{ij}]$ is the adjacency matrix of the graph. Recall that $a_{ii}=0$, $a_{ij}\in\real_{>0}$ if $j\in\VV$ can send information to agent $i\in\VV$, and zero otherwise. Moreover, a graph is undirected if the connection between the nodes is bidirectional and $a_{ij}=a_{ji}$ if $(i,j)\in\mathcal{E}$. Finally, an undirected graph is connected if there is a path from every agent to every other agent in the network (see e.g.~Fig.~\ref{fig:network}). Here, $\vect{L}=\text{Diag}(\vectsf{A}\vect{1}_N)-\vectsf{A}$ is the Laplacian matrix of the graph $\mathcal{G}$. The Laplacian matrix of a connected undirected graph is a symmetric positive semi-definite matrix that has a simple $\lambda_1=0$ eigenvalue, and the rest of its eigenvalues satisfy $\lambda_1=0< \lambda_2\leq\cdots\leq\lambda_N$.  Moreover, $\vect{L}\vect{1}_N=\vect{0}$. Since $\vect{L}$ of a connected undirected graph is a symmetric and real matrix, its normalized eigenvectors $v_1=\frac{1}{\sqrt{N}}\vect{1}_N,v_2,\cdots,v_N$ are mutually orthogonal. Moreover for  \begin{align}\label{eq::T}\vect{T}=\begin{bmatrix}\frac{1}{\sqrt{N}}\vect{1}_N&\vect{R}\end{bmatrix},\quad \vect{R}=\begin{bmatrix}v_2&\cdots&v_N\end{bmatrix}
\end{align}
we have $\vect{T}^\top\vect{L}\vect{T}=\vect{\Lambda}=\Diag{0,\lambda_2,\cdots,\lambda_N}$. We note that for any $\vect{q}\in\real^N$, we have $\|\vect{R}^\top\vect{q}\|=\|(\vect{I}_N-\frac{1}{N}\vect{1}_N\vect{1}_N^\top)\,\vect{q}\|$.
%\margin{this what we are using.}

%$\vect{R}\vect{R}^\top=(\vect{I}_N-\frac{1}{N}\vect{1}_N\vect{1}_N^\top)$, and thus $\|\vect{R}\|=\|\vect{I}_N-\frac{1}{N}\vect{1}_N\vect{1}_N^\top\|=1$.\margin{double check this}\hossein{I think $\|\vect{I}_N-\frac{1}{N}\vect{1}_N\vect{1}_N^\top\|=|\vect{R}\|=1$}

\section{Problem Definition}\label{sec::Prob_formu}
We consider a group of $N$ agents %integrator dynamics $\dot{x}^i=u^i$ 
each endowed with a one-sided time-varying 
%one-sided input 
%\begin{align}\label{eq::input_signal}
%\mathsf{r}^i(t)=\begin{cases}{r}^i(t),&t\in\realnonnegative\\
%0,&t\in\real_{<0},
%\end{cases}\quad i\in\VV,
%\end{align}
%where $r^i:\real_{\geq0}\to\real$ is a 
measurable locally essentially bounded signal $\mathsf{r}^i:\real_{\geq0}\to\real$, interacting over a connected undirected graph $\mathcal{G}(\mathcal{V},\mathcal{E},\vectsf{A})$. To obtain the average of their reference inputs, $\mathsf{r}^{\text{avg}}(t)=\frac{1}{N}\sum_{i=1}^N\mathsf{r}^i(t)$, these agents implement the distributed algorithm  
\begin{align}
   & \dot{x}^i(t)=-\alpha\,\sum\nolimits_{j=1}^N\!\!a_{ij}(x^j(t)\!-\!x^i(t))+\dot{\mathsf{r}}^i,~~i\in\VV,\label{eq::consensus-orig}\\ &x^i(0)=\mathsf{r}^i(0),\nonumber
\end{align}
where $\alpha\in\real_{>0}$. When the reference inputs of the agents are all static, i.e., $\dot{\mathsf{r}}^i=0$ for all $i\in\VV$,~\eqref{eq::consensus-orig} becomes the well-known Laplacian static average consensus algorithm that converges exponentially to $\mathsf{x}^{\text{avg}}(0)=\mathsf{r}^{\text{avg}}=\frac{1}{N}\sum_{j=1}^N \mathsf{r}^j$,  with the rate of convergence  $\rho_0=\alpha\lambda_2$ (for details see~\cite{ROS-JAF-RMM:07}). When one or more of the input signals are time-varying,~\eqref{eq::consensus-orig} is the dynamic average consensus algorithm of~\cite{dps-ros-rmm:05b}.%\footnote{Algorithm~\eqref{eq::consensus-orig},  can also be implemented in the alternative for of ($x^i(t)\!=y^i(t)+\mathsf{r}^i(t)$, $ \dot{y}^i(t)\!=-\alpha\,\sum\nolimits_{j=1}^N\!\!a_{ij}(x^j(t)\!-\!x^i(t))$, $y^i(0)=0$,~$i\in\VV$),
%see~\cite{SSK-BVS-JC-RAF-KML-SM} for more details. }
The convergence guarantee of~\eqref{eq::consensus-orig} is as follows.
\begin{thm}[Convergence of~\eqref{eq::consensus-orig} over an undirected connected graph~\cite{SSK-BVS-JC-RAF-KML-SM}]\label{thm::consensus-orig-conv}
{\rm Let $\GG$ be a connected undirected graph. Let $\|(\vect{I}_N-\frac{1}{N}\vect{1}_N\vect{1}_N^\top)\dvectsf{r}\|_{\infty}=\gamma<\infty$.  Then, for
  any $\alpha\in\realpositive$, the trajectories of algorithm~\eqref{eq::consensus-orig} are bounded and satisfy
\begin{equation}\label{eq::Alg_D_ultimate_bound} 
    \lim_{t\to\infty} \Big| x^i(t)-\mathsf{r}^{\text{avg}}(t) \Big| \leq 
    \eps_0,\quad\quad i\in\VV,
  \end{equation}
  where $\eps_0=\frac{\gamma}{ 
      \rho_{0}}$ and $\rho_0=\alpha\lambda_2$.  Moreover, The rate of convergence to this error neighborhood is no worse than $\rho_{0}$.}
\end{thm}
In this paper, with the intention of using outdated information to accelerate the convergence, we alter the average consensus algorithm~\eqref{eq::consensus-orig} to (compact representation)
\begin{subequations}\label{eq::consensus}
\begin{align}
%&x^i(t)=y^i(t)+\mathsf{r}^i(t) ,\label{eq::consensus_x}\\
&\dvect{x}(t)=-\alpha\,(1-\mathsf{k})\,\vect{L}\,\vect{x}(t)-\alpha\,\mathsf{k}\,\vect{L}\,\vect{x}(t-\tau)+\dvectsf{r},\label{eq::dyn_consensus_doty}\\
& x^i(0)=\mathsf{r}^i(0),~~
x^i(\eta)=0~\text{for}~\eta\in[-\tau,0),\quad i\in\VV,\label{eq::dyn_consensus_y_init}
\end{align}
\end{subequations}
 for $t\in\real_{\geq0}$, where $\mathsf{k}\in\real$ and $\tau\in\real_{\geq0}$. For $\mathsf{k}=0$,~\eqref{eq::consensus} recovers the original algorithm~\eqref{eq::consensus-orig}. We refer to  $\mathsf{k}$ as \emph{split factor}. 

To simply analyzing the convergence properties of~\eqref{eq::consensus}, we implement the change of variable  %$\vect{z}(t)=\vect{T}^\top(\vect{x}(t)-\mathsf{x}^\text{avg}(0)\vect{1}_N)$
\begin{align}\label{eq::Def_trans}
\vect{z}(t)=\vect{T}^\top(\vect{x}(t)-\mathsf{r}^\text{avg}(t)\vect{1}_N)
\end{align}
(recall~\eqref{eq::T}) to write~\eqref{eq::consensus} in equivalent form
\begin{subequations}\label{eq::LaclacianEquivalent}
\begin{align}
    \dot{z}_1(t)&=0,\quad\quad\quad {z}_1(0)=0, \label{eq::laclacian_equivalent_a}\\
\dvect{z}_{2:N}(t)&=-\alpha(1-\mathsf{k})\Bvect{\Lambda}\,\vect{z}_{2:N}(t) -\alpha\mathsf{k}\,\Bvect{\Lambda} \vect{z}_{2:N}(t-\tau)\nonumber\\
&~\quad+\vect{R}^\top\dvectsf{r}(t),\label{eq::LaclacianEquivalent_b}\\
\vect{z}_{2:N}(0)&=\vect{R}^\top\vectsf{r}(0), ~ \vect{z}_{2:N}(\eta)=\vect{0}~\text{for}~ \eta\in[-\tau,0), \label{eq::LaclacianEquivalent_bb}
\end{align}
\end{subequations}
where $\Bvect{\Lambda}=\Diag{\lambda_2,\cdots,\lambda_N}$. Under the given initial condition, the tracking error then is  
\begin{align}\label{eq::tracking_error}
\vect{x}(t)-\mathsf{r}^\text{avg}(t)\vect{1}_N=\vect{R}\,\vect{z}_{2:N}(t),\quad t\in\real_{\geq0}.
\end{align}
Using the method that specifies the solution of linear time-delayed systems~\cite{SY-PWN-AGU:10}, the  trajectory of~\eqref{eq::LaclacianEquivalent_b} under initial condition~\eqref{eq::LaclacianEquivalent_bb} is%~\cite{SY-PWN-AGU:10} 
\begin{align}\label{eq::Z_2N_sol}
 \vect{z}_{2:N}(t)\!=\!\sum_{j\in\mathbb{Z}}\!\ee^{\vect{S}_jt}\vect{C}_j\,\vect{z}_{2:N}(0)\!+\!\!\int_{0}^{t}\!\sum_{j\in\mathbb{Z}}\!\ee^{\vect{S}_j(t-\zeta)}\vect{C}_j^{'}\vect{R}^\top\dot{\vectsf{r}}(\zeta)\text{d}\zeta,   
\end{align}
where %\margin{if I remember it right, $C_j$ and $C_j'$ are equal for the initial condition we have and the way we define $C_j$ below.\hossein{You're right. It only holds for the scalar case. Here, for diagonal case we are good.}}
\begin{subequations}\label{eq::coefficients_infseries}
\begin{align}
\vect{S}_j&=\Diag{\vect{S}^1_j,\dots,\vect{S}_j^{N-1}},\label{eq::S_j}\\
\vect{S}^i_j&\!=\!\frac{1}{\tau}\vect{W}_j(-\alpha\mathsf{k}\lambda_{i+1}\tau\ee^{\alpha(1-\mathsf{k})\lambda_{i+1}\tau})\!-\!\alpha(1-\mathsf{k}){\lambda}_{i+1}\label{eq::S_j^i}\\
\vect{C}_j&=\Diag{\vect{C}_j^1,\dots,\vect{C}_j^{N-1}},\\
~\vect{C}_j^i&=\frac{1}{1-\alpha\mathsf{k}\lambda_{i+1}\tau\ee^{-\vect{S}_j^i\tau}},
\end{align}
\end{subequations}
and $\vect{C}_j'=\vect{C}_j$ because of the given initial conditions.%~\eqref{eq::dyn_consensus_y_init}.
%\end{subequations}
%\begin{subequations}
%\begin{align}
%\vect{S}_j&=\frac{1}{\tau}\vect{W}_j(-\alpha\,\mathsf{k}\,\Bvect{\Lambda}\tau)-\alpha(1-\maths{k})\Bvect{\Lambda},\label{eq::S_j}\\
%\vect{C}_j&=\Diag{\vect{C}_j^1,\dots,\vect{C}_j^{N-1}}, \vect{C}_j^i=\frac{1}{1+\alpha\mathsf{k}\lambda_{i+1}\tau\ee^{\vect{S}_j^i}}
%\end{align}
%\end{subequations}
%and the coefficients $\vect{C}_j$ and $\vect{C}_j^'$ can be determined numerically~\cite{SY-PWN-AGU:10}. 
%As we show below in Theorem~\ref{thm::Alg_delay}, 
When the reference input signals satisfy the condition given in Theorem~\ref{thm::consensus-orig-conv}, it follows from~\eqref{eq::Z_2N_sol} that $\vect{z}_{2:N}$ in the admissible delay range should converge exponentially to some neighborhood of zero, whose size is proportional to $\gamma$.  Moreover, the rate of convergence of algorithm~\eqref{eq::consensus} is $\rho_{\tau}(\mathsf{k})={\min}\{\{-\re{\vect{S}^i_j}\}_{i=1}^{N-1}\}_{j=-\infty}^{\infty}$. By invoking Lemma~\ref{lem::LambertReal}, $\rho_{\tau}(\mathsf{k})$  simplifies to $\rho_{\tau}(\mathsf{k})={\min}\{-\re{\vect{S}^i_0}\}_{i=1}^{N-1}$, which reads as
\begin{align}\label{eq::rate_consensus_d1}
    \rho_{\tau}(\mathsf{k})=\min\Big\{-\re{\frac{1}{\tau}&\vect{W}_0(-\alpha\mathsf{k}\lambda_i\tau\ee^{\alpha(1-\mathsf{k})\lambda_i\tau})}\nonumber\\
    &~~\qquad\qquad+\alpha(1-\mathsf{k}){\lambda}_i\Big\}_{i=2}^N. 
\end{align}
 For further discussion about the convergence rate of linear time delayed systems see~\cite[Corollary~1]{SD-JN-AGU:11}. 

%\begin{thm}[Convergence of~\eqref{eq::consensus} over an undirected connected digraph for delays in admissible range]\label{thm::consensus-delay}
%{\rm Let $\GG$ be a connected undirected graph. Let $\|(\vect{I}_N-\frac{1}{N}\vect{1}_N\vect{1}_N^\top)\dvectsf{r}\|_{\infty}=\gamma<\infty$ and $\tau$ be in admissible delay range.  Then, for
 % any $\alpha\in\realpositive$, the trajectories of algorithm~\eqref{eq::consensus-orig} are bounded and satisfy
%\begin{equation}\label{eq::Alg_D_ultimate_bound} 
%    \lim_{t\to\infty} \Big| x^i(t)-\mathsf{r}^{\text{avg}}(t) \Big| \leq 
%    \eps_\tau,\quad\quad i\in\VV,
%  \end{equation}
%  where $\eps_\tau=\frac{\gamma\,\kappa_\tau}{ 
%      \rho_{\tau}}$ and 
%      \begin{align}
%      \rho_\tau\end{align*}. Moreover, The rate of convergence to this error neighborhood is no worse than $\rho_{\tau}$.}
%\end{thm}

 Our objective in this paper is to show that by splitting the disagreement feedback into a current $-\alpha\,(1-\mathsf{k})\,\vect{L}\,\vect{x}(t)$ and an outdated $-\alpha\,(1-\mathsf{k})\,\vect{L}\,\vect{x}(t-\tau)$ components, it is possible to increase the rate of convergence of  algorithm~\eqref{eq::consensus}. Specifically, we determine for what values of $\mathsf{k}$, there exists ranges of time delay that the rate of convergence of~\eqref{eq::consensus} increases (ranges of delay for which decay rate of the transient response of~\eqref{eq::consensus} increases). We also specify the maximum reachable rate due to delay and its corresponding maximizer delay. One may argue that the rate of convergence of~\eqref{eq::consensus-orig} can be increased by `cranking up' the gain $\alpha$. However, this choice leads to increase in the control effort of the agents. In our study, then, we set to identify values of split factor $\mathsf{k}$ for which for a fixed $\alpha$ the increase in the convergence rate of~\eqref{eq::consensus} due to delay  in comparison to~\eqref{eq::consensus-orig} is without increasing the control effort. Finally, we prove that for delays in the admissible delay bound, the ultimate tracking error of~\eqref{eq::consensus} is the same as~\eqref{eq::Alg_D_ultimate_bound}. This assertion, increases the appeal of the modified average consensus algorithm~\eqref{eq::consensus} as an effective algorithm that yields faster convergence than the original algorithm~\eqref{eq::consensus-orig}.  We close this section by noting that following the change of variable method proposed in~\cite{SSK-BVS-JC-RAF-KML-SM}, algorithm~\eqref{eq::dyn_consensus_doty} can be implemented in the alternative way (recall that $\mathsf{r}^i$ is a one-sided signal)%\margin{\hossein{This form is fine and is in accordance with (7). But we can write it in more general form where $\sum_{i=1}^{N}y^i(0)=0$}\solmaz{lets keep it the way it is}}
\begin{align*}
&\dot{y}^i(t)=-\alpha\,(1-\mathsf{k})\,\vect{L}\,\vect{x}(t)-\alpha\,\mathsf{k}\,\vect{L}\,\vect{x}(t-\tau),\\
&x^i(t)=y^i(t)+\mathsf{r}^i(t),\\
&y^i(0)=0,\,\,x^i(\eta)=0~\text{for}~\eta\in[-\tau,0),\quad i\in\VV,
\end{align*}
which does not require knowledge of derivative of the reference input of the agents . 

\section{Accelerating average consensus using outdated feedback}\label{sec::main}
In this section, we study the effect of the outdated feedback on the convergence rate and the ultimate tracking response of the modified average consensus algorithm~\eqref{eq::consensus}. 
%in both static and dynamic forms. 
% Here, we use our preliminary developments in Section~\ref{sec::Prelim} to construct our results.

To start our study, we identify the admissible delay range $(0,\bar{\tau})$ for algorithm~\eqref{eq::consensus} for different values of split factor $\mathsf{k}$. Given the tracking error~\eqref{eq::tracking_error}, the admissible delay bound is determined by the ranges of delay for which the zero input dynamics of~\eqref{eq::LaclacianEquivalent_b} preserves its exponential stability.

\begin{comment}
As we saw in the arguments leading to~\eqref{eq::rate_consensus_d1}, the rate of convergence is determined by the rate of decay of $\ee^{\vect{S}_jt}$. Therefore, the rate of convergence of~\eqref{eq::consensus} can be obtained  by studying internal stability of the zero-input dynamics of~\eqref{eq::LaclacianEquivalent_bb}.
\end{comment}

%%%%%%%%%%%%%%%%%%%%%%%%%%%%%%%%%%%%%%%%%%%%%%%%
%%%%%%%%%%%%%%%%%%%%%%%%%%%%%%%%%%%%%%%%%%%%%%%%
\begin{lem}[%Admissible range of delay for~\eqref{eq::laclacian_equivalent_zeroinput} %
Admissible rage of delay for internal stability of algorithm~\eqref{eq::consensus}]\label{lem::admis_delay_modified_con}
{\rm The following assertions hold for the modified average consensus algorithm~\eqref{eq::consensus} over an undirected connected graph (recall~\eqref{eq::S_j^i}).   
\begin{itemize}
    \item[(a)] For  $\mathsf{k}\leq 0.5$, the modified average consensus  algorithm~\eqref{eq::consensus} is internally stable  for any $\tau\in\real_{\geq0}$,  i.e., $\bar{\tau}=\infty$. 
    %Moreover, for any $\tau\in\real_{\geq0}$, we have $\lim_{t\to\infty} \ee^{\vect{S}_j^it}=0$, $i=\{1,\cdots,N-1\}$ and $j\in\mathbb{N}$ , i.e., $\bar{\tau}=\infty$.
   % Moreover, under the stated initial conditions, zero-input dynamics of algorithm~\eqref{eq::consensus} converges exponentially fast to $\mathsf{x}^{\text{avg}}(0)$  for any $\tau\in\real_{\geq0}$.
     \item[(b)] 
     For  $\mathsf{k}>0.5$, the modified average consensus algorithm~\eqref{eq::consensus} is internally stable if and only if $\tau\in[0,\bar{\tau})$, where 
    \begin{align}\label{eq::consensus_delay_bounded}
        %\bar{\tau}=\frac{\arccos(-\frac{1}{\kappa})}{\lambda_N\sqrt{\kappa^2-1}}.
         \bar{\tau}=\arccos(1-1/\mathsf{k})/(|\alpha|\lambda_N\sqrt{2\mathsf{k}-1}).
    \end{align}
   % \item[(c)] for $\mathsf{k}>1$ the modified average consensus algorithm~\eqref{eq::consensus} is unstable regardless of the value of $\tau\!\in\!\real_{>0}$.
    %\hossein{we can remove this part if we just consider 
\end{itemize}
Also, for any $\tau\in[0,\bar{\tau})$, we have  $\lim_{t\to\infty} \ee^{\vect{S}_j^it}=0$, $i\in\mathbb{Z}_1^{N-1}$ and $j\in\mathbb{Z}$. Moreover, under the initial condition~\eqref{eq::dyn_consensus_y_init}, the trajectories of $x^i$, $i\in\VV$ of the zero-input dynamics of algorithm~\eqref{eq::consensus} converges exponentially fast to $\mathsf{x}^{\text{avg}}(0)$. }
\end{lem}
\begin{proof}
Consider the zero-input dynamics of~\eqref{eq::LaclacianEquivalent}, the equivalent representation of zero dynamics of algorithm~\eqref{eq::consensus}.  
It is evident that the delay tolerance of~\eqref{eq::LaclacianEquivalent} is defined by the dynamics of states $\vect{z}_{2:N}$. Note that~\eqref{eq::LaclacianEquivalent_b} because of definition of $\Bvect{\Lambda}$ reads also as
%\begin{subequations}
\begin{align}\label{eq::laclacian_equivalent}
\dot{z}_i(t)&\!=\!-\alpha(1-\mathsf{k})\lambda_i\,z_{i}(t) -\alpha\mathsf{k}\,\lambda_i z_{i}(t-\tau)%+[\vect{R}^\top\dvectsf{r}(t)]_i
,~~~i\in\mathbb{Z}_{2}^N.
%\\z_{i}(0)&=[\vect{T}^\top\vectsf{r}(0)]_i, \quad z_{i}(\eta)=0~~~ \eta\in[-\tau,0), \label{eq::laclacian_equivalent_y2}
\end{align}
%\end{subequations}
%where $i\in\{2,\cdots,N\}$. 
When $\mathsf{k}\leq0.5$ we have $-\alpha\lambda_i(1-\mathsf{k})\leq|\alpha\lambda_i\mathsf{k}| $, while when $\mathsf{k}>0.5$ we have $-\alpha\lambda_i\mathsf{k}< -|\alpha\lambda_i(1-\mathsf{k})|$. Therefore, the admissible delay ranges stated in the statement (a)
and the statement (b) follow, respectively, from the statements (a) and (b) of Lemma~\ref{thm::stability_con}.  To establish~\eqref{eq::consensus_delay_bounded}, we used $\bar{\tau}=\min\{\bar{\tau}_i\}_{i=2}^N$ where according to~\eqref{eq::Admis_range_tau} we have
\begin{align}\label{eq::tau_i}
    \bar{\tau}_i=\frac{\arccos(1-\frac{1}{\mathsf{k}})}{\alpha\lambda_i\sqrt{2\mathsf{k}-1}}.
\end{align}
In admissible delay bound, the time-delayed systems~\eqref{eq::laclacian_equivalent} for $i\in\{2,\cdots,N\}$ are exponentially stable, i.e.,  $z_i\rightarrow0$ as $t\to{\infty}$, $i\in\{2,\cdots,N\}$. As a result, $\lim_{t\to\infty} \ee^{\vect{S}_j^it}=0$, $i\in\mathbb{Z}_1^{N-1}$ and $j\in\mathbb{Z}$ can be certified from~\eqref{eq::Z_2N_sol} when the second term in the right-hand side is removed (zero-input response). Moreover since 
$\vect{z}(t)=\vect{T}^\top\vect{x}(t)$
%\margin{double check this\hossein{\\This is not true.$\vect{T}^\top\mathsf{r}^\text{avg}\vect{1}_N\neq0$}
%we are analyzing the zero input dynamics, I just want to make sure we do not need to write this as $\vect{z}(t)=\vect{T}^\top(\vect{x}(t)-\mathsf{x}^\text{avg}(0)\vect{1}_N)$}
(in zero-input dynamics), we then obtain that in the stated admissible delay ranges in the statements (a) and (b), $\vect{x}(t)$ converges exponentially fast to $\frac{1}{\sqrt{N}}z_1(0)\vect{1}_N=\frac{1}{\sqrt{N}}(\frac{1}{\sqrt{N}}\sum_{j=1}^N x^i(0))\vect{1}_N=\mathsf{x}^{\text{avg}}(0)$. %Under the stated initial conditions, then $x^i$, $i\in\VV$ converges to $\mathsf{x}^{\text{avg}}(0)$ exponentially fast. 
This completes the proof. % \margin{why initial condition matters here.}
\end{proof}
The results of Lemma~\ref{lem::admis_delay_modified_con} includes the result in~\cite{HM-SSK:18}, which specifies the admissible range of delay for when $\mathsf{k}=1$, as special case.
Next, we study the ultimate tracking bound of the modified average consensus algorithm~\eqref{eq::consensus}. We show that for delays in the admissible delay bound the ultimate tracking error is still $\eps_0$ as defined in Lemma~\ref{lem::admis_delay_modified_con}. %The confidence that delay does not increase the tracking error bound makes the modif 

\begin{thm}[Convergence of~\eqref{eq::consensus} over  connected graphs when $\tau\in[0,\bar{\tau})$ ]\label{thm::Alg_delay}
%\solmaz{should be written properly to cover both $\mathsf{k}=0$ and $1$}
{\rm Let $\GG$ be a connected undirected graph with communication delay in $\tau\in[0,\bar{\tau})$ where $\bar{\tau}$ is specified in Lemma~\ref{lem::admis_delay_modified_con}. Let $\|(\vect{I}_N-\frac{1}{N}\vect{1}_N\vect{1}_N^\top)\dvectsf{r}\|_{\infty}=\gamma<\infty$.  Then, for
  any $\alpha\in\realpositive$, the trajectories of algorithm~\eqref{eq::consensus} for any $\mathsf{k}\in\real$  are bounded and satisfy~\eqref{eq::Alg_D_ultimate_bound}.
%\begin{equation}\label{eq::Alg_D_ultimate_bound1} 
%    \lim_{t\to\infty} \Big| x^i(t)-\mathsf{r}^{\text{avg}}(t) \Big| \leq 
%    \eps_0,
%  \end{equation}
%  for any $i\in\VV$ where
%  $\eps_{0}=\frac{\gamma}{\rho_0}$ with $\rho_0=\mathsf{b}\lambda_2$.
  %$\eps_{\tau}=\gamma\max\{|\sum_{j=-\infty}^{j=+\infty}(\frac{C_j^i}{s_j^i})|\}_{i=2}^N$ with $C_j^i=\frac{1}{1-\alpha\mthsf{k}\lambda_i\tau\ee^{-s^i_j\tau}}$ and $s^i_j=\frac{1}{\tau}W_j(-\alpha\mathsf{k}\lambda_i\tau\ee^{\alpha(1-\mathsf{k})\lambda_i\tau})-\alpha(1-\mathsf{k})\lambda_i$ for $i\in\{2,\cdots,N\}$.
  Moreover, The rate of convergence to this error neighborhood is no worse than~$\rho_\tau(\mathsf{k})$ defined in~\eqref{eq::rate_consensus_d1}.}
\end{thm}
\begin{proof}
%\margin{This proof should be revised since some of the material here are discussed in the earlier parts of the paper}
To establish our proof we consider~\eqref{eq::LaclacianEquivalent}, the equivalent representation of algorithm~\eqref{eq::consensus}. Recall~\eqref{eq::laclacian_equivalent_a} which along with the given initial condition gives $z_1(t)=0$  %$z_1(t)=\frac{1}{\sqrt{N}}\sum_{i=1}^N\mathsf{r}^i(t)-\frac{1}{\sqrt{N}}\sum_{i=1}^N\mathsf{r}^i(0)$ \solmaz{??}
for $t\in\real_{\geq0}$. Also, given~\eqref{eq::Z_2N_sol}, the trajectories of $t\mapsto\vect{z}_{2:N}$ for $t\in\real_{\geq0}$ satisfy 
\begin{align}\label{eq::z2Nnormbound} 
&\!\|\vect{z}_{2:N}(t)\|\!\!\leq\!\left\|\sum_{j\in\mathbb{Z}}\!\Diag{\vect{C}^1_{j}\ee^{\vect{S}^1_jt},\cdots,\vect{C}^{N-1}_{j}\ee^{\vect{S}^{N-1}_jt}}\right\|\!\left\|\vect{z}_{2:N}(0)\right\|\nonumber\\
&\!
 \,\,\,+\gamma\left\|\sum_{j\in\mathbb{Z}}\Diag{\frac{\vect{C}^1_{j}}{\vect{S}^1_j}(1-\ee^{\vect{S}^1_jt}),\cdots,\frac{\vect{C}^{N-1}_{j}}{\vect{S}^{N-1}_j}(1-\ee^{\vect{S}^{N-1}_jt)}}\right\|=\nonumber\\
 &\,\,\max\Big\{\big|\!\sum_{j\in\mathbb{Z}}\!\vect{C}^i_{j}\ee^{\vect{S}^i_jt}\big|\Big\}_{i=1}^{N-1}\left\|\vect{z}_{2:N}(0)\right\|+\nonumber\\&\,\,\,\gamma\max\Big\{\big|\sum_{j\in\mathbb{Z}}\frac{\vect{C}^i_{j}}{\vect{S}^i_j}(1-\ee^{\vect{S}^i_jt})\big|\Big\}_{i=1}^{N-1}.
\end{align}
Here, we used  $\|\vect{R}^\top\dvectsf{r}\|\leq\gamma$. %\margin{shouldn't his be $\|\vect{R}^\top\dvectsf{r}\|\leq\gamma$} 
 Furthermore, using~\eqref{eq::Def_trans} we obtain $|x^i(t)-\mathsf{r}^\text{avg}(t)|\leq\|\vect{x}(t)-\mathsf{r}^\text{avg}(t)\vect{1}_N\|=\|\vect{z}(t)\|=\sqrt{|z_1(t)|^2+\|\vect{z}_{2:N}(t)\|^2}=\|\vect{z}_{2:N}(t)\|$.
 %\solmaz{this should be $|x^i(t)-\mathsf{r}^\text{avg}(t)|\leq\|\vect{x}(t)-\mathsf{r}^\text{avg}(t)\vect{1}_N\|=\|\vect{z}(t)\|=\sqrt{|z_1(t)|^2+\|\vect{z}_{2:N}(t)\|^2}=\|\vect{z}_{2:N}(t)\|$}
 Then, it follows from~\eqref{eq::z2Nnormbound} that $\lim_{t\to\infty}| x^i(t)-\mathsf{r}^{\text{avg}}(t) |\leq\lim_{t\rightarrow\infty}\|\vect{z}_{2:N}(t)\|=\gamma\max\{|\sum_{j\in\mathbb{Z}}\frac{\vect{C}^i_{j}}{\vect{S}^i_j}|\}_{i=1}^{N-1}$. Next, we show that $\sum_{j\in\mathbb{Z}}\frac{\vect{C}^i_{j}}{\vect{S}^i_j}=\frac{1}{\alpha\lambda_{i+1}}$ for any $i\in\mathbb{Z}_1^{N-1}$.
 To this end, note that from zero-input response of~\eqref{eq::Z_2N_sol} we have ${z}_{i}(t)\!=\!(\sum_{j\in\mathbb{Z}}\!\ee^{\vect{S}^i_jt}\vect{C}^i_j)\,{z}_{i}(0)$ which gives ($\sum_{j\in\mathbb{Z}}\frac{\vect{C}^i_{j}}{\vect{S}^i_j}){z}_{i}(0)=\int_{0}^{\infty}{z}_{i}(t)\text{d}t$.  On the other hand using~\eqref{eq::laclacian_equivalent} for any $i\in\mathbb{Z}_1^{N-1}$ we have
 \begin{align*}
    \int_0^\infty\!\!\!\dot{z}_{i+1}(t)\text{d}t&=\!-\alpha\lambda_{i+1} \int_0^\infty\!\!\!\! z_{i+1}(t)\text{d}t-\alpha\lambda_{i+1} \mathsf{k}\!\int_{-\tau}^0\!\!\! z_{i+1}(t)\text{d}t.
 \end{align*}
Recalling~\eqref{eq::LaclacianEquivalent_bb}, we get $\int_{-\tau}^0 z_{i+1}(t)\text{d}t=0$, which along with the fact that under admissible range $\lim_{t\rightarrow\infty}z_{i+1}(t)=0$, implies $\int_0^\infty z_{i+1}(t)\text{d}t=\frac{z_{i+1}(0)}{\alpha\lambda_i}$ which holds for any initial condition $z_{i+1}(0)\in\real$. Therefore, we get $\sum_{j\in\mathbb{Z}}\frac{\vect{C}^i_{j}}{\vect{S}^i_j}=\frac{1}{\alpha\lambda_{i+1}}$ and consequently $\lim_{t\to\infty}| x^i(t)-\mathsf{r}^{\text{avg}}(t) |\leq \frac{\gamma}{\alpha\min\{\lambda_i\}_{i=2}^N}\leq\frac{\gamma}{\rho_0}$.
Moreover, the maximum rate of convergence corresponds to the worst rate of the exponential terms in~\eqref{eq::z2Nnormbound}, or equivalently  ${\min}\{\{-\re{\vect{S}^i_j}\}_{i=1}^{N-1}\}_{j\in\mathbb{Z}}$  given in~\eqref{eq::rate_consensus_d1}.
%,which concludes the~proof.
\end{proof}

So far we have shown that splitting the immediate disagreement feedback of~\eqref{eq::consensus-orig} into current and outdated components as in~\eqref{eq::consensus} does not have adverse effect on the tracking performance. Next, we show that this action interestingly can lead to increase in the rate of the converge at some specific values of $\mathsf{k}$ and $\tau$. As we noted earlier, the rate of convergence of~\eqref{eq::consensus} is determined by behavior of its transient response that is governed by its zero-input dynamics. Consequently, we study the  stability of the zero-input dynamics of the modified average consensus algorithm~\eqref{eq::consensus} and examine how its exponential rate of convergence  to the average of its initial condition at time $t=0$ changes due to delay at various values of $\mathsf{k}\in\real/\{0\}$. For any given value of $\mathsf{k}$ and $\tau$, in what follows, we let $\rho_{\tau}(\mathsf{k})$ be the rate of convergence of~\eqref{eq::consensus} and  $\vect{u}_{\tau,\mathsf{k}}(t)=-\alpha\,(1-\mathsf{k})\,\vect{L}\,\vect{x}(t)-\alpha\,\mathsf{k}\,\vect{L}\,\vect{x}(t-\tau)$ be the control effort to steer the zero-input dynamics of~\eqref{eq::consensus}.
%Our focus is to show that for any $\mathsf{k}\in\realpositive$, there exists a range of delay for which the convergence rate of algorithm~\eqref{eq::consensus} improves. Also, we demonstrate that for any $\mathsf{k}\in(0,1]$ higher convergence rate is achieved without increasing the control effort.
Specifically, we show that for all $\mathsf{k}\in\realpositive$, there always exists a range of delay $(0,\tilde{\tau}_{\mathsf{k}})$ such that $\rho_{\tau}(\mathsf{k})>\rho_{0}(0)=\rho_0=\alpha\lambda_2$  for any $\tau\in(0,\tilde{\tau}_\mathsf{k})$. 
%Here $\rho_{\tau,\mathsf{k}}$ is the rate of convergence of the modified algorithm~\eqref{eq::consensus} in the presence of delay for a given $\mathsf{k}\in\realpositive$ and $\rho_0\!=\!\alpha\lambda_2$ is the rate of convergence of algorithm~\eqref{eq::consensus-orig}. 
We show however that only for $\mathsf{k}\in(0,1]$ we can guarantee  $|\vect{u}_{{\tau},\mathsf{k}}|_{\infty}\leq |\vect{u}_{0,0}|_{\infty}$, for $\tau\in(0,\tilde{\tau}_{\mathsf{k}})$. %\margin{double check this statement, if you agree remove the next one.}
%\st{Moreover, for any $\mathsf{k}\in(0,1]$, we show that $\|\vect{u}_{\tau,\mathsf{k}}\|_{\infty}\leq \|\vect{u}_\mathsf{0,0}\|_{\infty}$, for $\tau\in(0,\tilde{\tau}_{\mathsf{k}})$.} 
In what follows, we also investigate what the maximum value of $\rho_{\tau}(\mathsf{k})$ and the corresponding maximizer $\tau^\star_{\mathsf{k}}\in(0,\tilde{\tau}_\mathsf{k})$ are for a given $\mathsf{k}\in\realpositive$. 

We start our analysis, by defining the \emph{delay gain function} 
\begin{align}\label{eq::gain_rate}
\!\!\!g(\gamma,\mathsf{x})=\begin{cases}
\frac{1}{\mathsf{x}}\re{W_0(\mathsf{x}\,\ee^{\gamma\,\mathsf{x}})},\quad\quad \mathsf{x}\in\real\backslash\{0\},\quad\quad \\ 1,  \quad\quad\quad\quad\quad\quad\quad\quad\quad\quad \mathsf{x}=0, \end{cases}
\end{align}
with $\mathsf{x},\gamma\in\real$, to write $\rho_\tau(\mathsf{k})$ in~\eqref{eq::rate_consensus_d1}  as
\begin{align}\label{eq::rate_consensus_d}
    &\rho_{\tau}(\mathsf{k})={\min}\{\rho_{\tau,i}(\mathsf{k})\}_{i=2}^N, \\ &\rho_{\tau,i}(\mathsf{k})=(\mathsf{k} g(1-\frac{1}{\mathsf{k}},-\mathsf{k}\lambda_i\alpha\tau)\!+\!(1-\mathsf{k}))\,\alpha\lambda_i.\nonumber
\end{align}
It follows from~\eqref{eq::limWz-z}~that $
\lim_{\mathsf{x}\to0}g(\gamma,\mathsf{x})=1$. Therefore, as expected, $\lim_{\tau\to0}\rho_\tau(\mathsf{k})=\rho_0=\alpha\lambda_2$. 
We note that in fact $\rho_{\tau,i}$, $i\in\{2,\cdots,N\}$, defines the rate of convergence of $z_i$ in~\eqref{eq::laclacian_equivalent}. In what follows, when emphasis on $\mathsf{k}$ is not necessary, to simplify the notation we write $\rho_{\tau}(\mathsf{k})$ as $\rho_{\tau}$.
 
 In Appendix~\ref{sec::Prelim}, we study the variation of delay gain function versus $\mathsf{x}\in\real_{\geq0}$ for given values of $\gamma$.
We show that for some specific values of ${\gamma}$ there always exists a subset of the admissible delay range that the delay gain $g(\gamma,\alpha\tau)$ is greater than $1$.
%, thus the delay can result in an increase in the rate of convergence (recall~\eqref{eq::rate_consensus_d}). 
In what follows, we use these results to determine ranges of delay and $\mathsf{k}$ we have $\rho_\tau(\mathsf{k})>\alpha\lambda_2$. We also identify the optimum value of the delay $\tau^\star$ for which $\rho_{\tau}$ has its maximum value, i.e., we identify the solution for 
\begin{align}\label{eq::ro_min_max}
 \tau^\star=\underset{\tau\in(0,\bar{\tau})}{\argmax}\,\rho_{\tau}=\underset{\tau\in(0,\bar{\tau})}{\argmax}\,\min \{\rho_{\tau,i}\}_{i=2}^N.
\end{align}
As we showed in Appendix~\ref{sec::Prelim}, the variation  of the delay gain function $g$  with $\mathsf{x}\in\real_{\geq0}$ for given values of $\gamma$ is not monotone. Therefore, the solution to~\eqref{eq::ro_min_max} is not trivial. Our careful characterization of variation of $g$ vs. $\mathsf{x}\in\real_{\geq0}$ in Appendix~\ref{sec::Prelim} however, let us achieve our goal. 

In what follows, we set
\begin{align*}
 &\tau_i^\star=\underset{\tau\in(0,\bar{\tau}_i)}{\argmax}\,\rho_{\tau,i},\\
 &\tilde{\tau}_i=\{\tau\in(0,\bar{\tau}_i)\,|\,g(1-\frac{1}{\mathsf{k}},-\mathsf{k}\alpha\lambda_i\tau))\!=\!1\}.
\end{align*}
where $\bar{\tau}_i$ is given in~\eqref{eq::tau_i}. With the notation defined, the next theorem examines the effect of outdated feedback on the rate of convergence of modified consensus algorithm~\eqref{eq::consensus} for different value of $\mathsf{k}\in\real/\{0\}$.
%%%
%%%
\begin{thm}[Effect of outdated feedback on the rate of convergence of average consensus algorithm~\eqref{eq::consensus}]\label{thm::consensus_prob}
{\rm The following assertions hold for the modified average consensus dynamics~\eqref{eq::consensus} over a connected graph whose rate of convergence is specified in~\eqref{eq::rate_consensus_d}:
\begin{itemize}
    \item[(a)] For $\mathsf{k}<0$ the rate of convergence of the consensus algorithm~\eqref{eq::consensus} decreases by increasing $\tau\in\real_{\geq0}$. 
\item[(b)] For $\mathsf{k}>0$, $\rho_\tau> \rho_0$ if and only if  $\tau\in[0,\hat{\tau})\subset[0,\bar{\tau})$ where $\hat{\tau}=\min\{\hat{\tau}_i\}_{i=2}^{N}$ with $\hat{\tau}_i=\{\tau\in\realpositive|\rho_{\tau,i}=\rho_0\}$ and satisfies $\tilde{\tau}_N\leq \hat{\tau}\leq \min\{\tilde{\tau}_2,\bar{\tau}\}$. 
%\margin{why do you use $\hat{\tau}$ why not $\tilde{\tau}$\hossein{$\hat{\tau}$ is where $\rho_\tau=\rho_0$. So we have  $\hat{\tau}_i\geq\tilde{\tau}_i$. the equality holds for $i=2$.}}
     %$\tilde{\tau}_N=\{\tau\in[0,\bar{\tau})\,|\,g(-\frac{1}{\kappa},-\lambda_N\tau)=1\}$, 
      Moreover,
    the optimum time delay $\tau^{\star}$ corresponding to the maximum rate of convergence of the consensus algorithm~\eqref{eq::consensus} satisfies
 $\tau^{\star}\in[\tau^\star_N,\min\{\tau^\star_2,\hat{\tau}\}]$, where $\tau_N^\star=\frac{1}{\alpha(1-\mathsf{k})\lambda_N}W_0(\frac{1-\mathsf{k}}{\mathsf{k}\ee})$ and $\tau_2^\star=\frac{1}{\alpha(1-\mathsf{k})\lambda_2}W_0(\frac{1-\mathsf{k}}{\mathsf{k}\ee})$, and is given by $\tau^\star=\{\tau\in[\tau^\star_N,\min\{\tau^\star_2,\hat{\tau}\}]\,|\, \rho_{\tau,2}=\min\{\rho_{\tau,i}\}_{i=3}^N\}$.
 \end{itemize}}
\end{thm}
\begin{proof}
%\margin{In this proof you need to change  $\kappa$ to $\mathsf{k}$}
Recall that the rate of convergence of algorithm~\eqref{eq::consensus} is specified by~\eqref{eq::rate_consensus_d} (equivalent representation of~\eqref{eq::rate_consensus_d1}), which is the minimum of the rate of convergence of $z_i$, $i\in\{2,\cdots,N\}$ dynamics given in~\eqref{eq::laclacian_equivalent}.
Then, the proof of part (a) follows directly from statement (a) of Theorem~\ref{thm::delay_effect_scalar}, which states that the rate of convergence of each $z_i$, $i\in\{2,\cdots,N\}$ dynamics 
 decreases by increasing delay $\tau\in\real_{>0}$ (note that in Theorem~  \ref{thm::delay_effect_scalar} each $z_i$ dynamics reads as $\mathsf{a}=-\alpha\mathsf{k}\lambda_i>0$
 %\margin{\hossein{We have two alpha}how about changing $\alpha$ and $\mathsf{b}$ to $\mathsf{a}$ and $\mathsf{b}$? we are using plain $a$ for adjacency matrix elements}
 and $\mathsf{b}=-\alpha(1-\mathsf{k})\lambda_i<0$). 
 To prove statement~(b) we proceed as follows. For $\mathsf{k}>0$, because of the statement (b) of Theorem~  \ref{thm::delay_effect_scalar} for each $z_i$, $i\!\in\!\{2,\cdots,N\}$, dynamics  ($\mathsf{a}\!=\!-\alpha\mathsf{k}\lambda_i\!<\!0$) we have the guarantees that 
 \begin{align*} \rho_{\tau,i}\!=\!(\mathsf{k} g(1-\frac{1}{\mathsf{k}},-\mathsf{k}\lambda_i\alpha\tau)\!+\!(1-\mathsf{k}))\,\alpha\lambda_i>\rho_{0,i}\geq \rho_0,
 \end{align*}
% \belowdisplayskip 
% For any $j\in\mathcal\{2,\cdots,N\}$ let  $\hat{\tau}_j\in\realpositive$ be the time delay where $\rho_{\hat{\tau}_j,j}=\rho_0$. Obviously $\hat{\tau}_2=\tilde{\tau}_2$. 
for $\tau\in(0,\tilde{\tau}_i)$. Since $\alpha>0$, $\lambda_N\geq \lambda_{N-1}\geq \cdots\geq \lambda_2>0$ and $\rho_{0,N}\geq\rho_{0,N-1}\geq \cdots\geq\rho_{0,2}$, we have 
 \begin{subequations}
 \begin{align}
    & \tau^\star_i<\tilde{\tau}_i\leq \hat{\tau}_i<\bar{\tau}_i,\quad i\in\{3,\cdots,N\},\label{eq::hat_i_bound}\\
    &\tilde{\tau}_N\leq\tilde{\tau}_{N-1}\leq \cdots\leq \tilde{\tau}_2,\label{eq::tilde_i_bound}\\
    &\bar{\tau}=\bar{\tau}_N\leq\bar{\tau}_{N-1}\leq \cdots\leq \bar{\tau}_2,\label{eq::bar_i_bound}\\
    &{\tau}_N^{\star}\leq{\tau}_{N-1}^{\star}\leq \cdots\leq {\tau}_2^{\star}.\label{eq::star_i_bound}
    %\\ & \rho_{0,2}\leq\cdots\leq\rho_{0,N},\label{eq::rho_0_i_relative}
 \end{align}
 \end{subequations}
 and $\hat{\tau}_2=\tilde{\tau}_2$. 
 Since $g(1-\frac{1}{\mathsf{k}},-\mathsf{k}\lambda_i\alpha\tau)$ is a decreasing function of $\tau$ for any $\tau\in(\tilde{\tau}_i,\bar{\tau}_i)\subset(\tau_i^\star,\bar{\tau}_i)$ (Recall Lemma~  \ref{thm::g(k,x)_prop}), it follows that for any $\tau\in[0,\hat{\tau}_j)$ we have $\rho_{\tau,j}>\rho_0$ and for any $\tau\in[\hat{\tau}_j,\bar{\tau})$ we have $\rho_{\tau,j}<\rho_0$. 
 Because $\rho_\tau=\min\{\rho_{\tau,j}\}_{j=2}^N$, we have $\rho_\tau>\rho_0$, if and only if $\tau\in(0,\hat{\tau})$ where $\hat{\tau}=\min\{\hat{\tau}_j\}_{j=2}^N$. From~\eqref{eq::hat_i_bound} and \eqref{eq::tilde_i_bound}, it follows that $\tilde{\tau}_N\leq \hat{\tau}$. Moreover, since $\rho_{\tau,2}>\rho_0$ for $\tau\in(0,\tilde{\tau}_2)$, we obtain $\hat{\tau}\leq \min\{\tilde{\tau}_2,\bar{\tau}\}$. 
 This concludes the proof of the first part of statement (b).
 
 To obtain $\tau^\star\in(0,\hat{\tau})$ which gives the maximum attainable $\rho^\star_\tau$ we proceed as follows.  First, note that statement (b) of Theorem~  \ref{thm::delay_effect_scalar} indicates that $\rho_{\tau,i}$, $i\in\{2,\cdots,N\}$ is a monotonically increasing (resp. decreasing) function of  $\tau\in(0,\tau^\star_i)$ (resp. $\tau\in(\tau^\star_i,\bar{\tau}_i)$). Then because of~\eqref{eq::star_i_bound}, we have the guarantees that $\rho_\tau$  is a monotonically increasing function of $\tau\in(0,\tau^\star_N)$, and decreasing function of $\tau$ for any  $\tau>\tau^\star_2$. Therefore, the maximum value of $\rho_\tau$ should be attained at $\tau^\star\in([\tau^\star_N,\tau^\star_2]\cap(0,\hat{\tau}))\subseteq[\tau^\star_N,\min\{\tau^\star_2,\hat{\tau}\}]$ with $\tau^\star_2=\frac{1}{\lambda_2}W_0(\frac{1-\mathsf{k}}{\mathsf{k}\ee})$ and $\tau^\star_N=\frac{1}{\lambda_N}W_0(\frac{1-\mathsf{k}}{\mathsf{k}\ee})$. Now let $\mathsf{j}=\min\{i\in\{2,\cdots,N\}|\tau^\star_i\leq\tau^\star\}$. Then, given~\eqref{eq::star_i_bound}, for any $i\in\{2,\cdots,N\}$ such that $i< \mathsf{j}$ (resp. $i\geq \mathsf{j}$) by virtue of statement (e) of Lemma~  \ref{thm::g(k,x)_prop}  we know $\text{d} g(1-\frac{1}{\mathsf{k}},-\mathsf{k}\lambda_i\alpha\tau)/\text{d}\tau>0 $ (resp. $<0$) and consequently $\text{d}\rho_{\tau,i}/\text{d}\tau>0$ (resp. $<0$) at $\tau=\tau^\star$. Since $\rho_\tau=\min\{\rho_{\tau,i}\}_{i=2}^N$, the maximum value of $\rho_\tau$ is attained at $\tau=\tau^\star$ at which
\begin{align}\label{eq::temp1}
\min\{\rho_{\tau,i}\}_{i=\mathsf{j}}^N=\min\{\rho_{\tau,i}\}_{i=2}^{\mathsf{j}-1}.
\end{align}
Since $\lambda_2\tau^\star\leq\cdots\leq\lambda_{\mathsf{j}-1}\tau^\star$ and $\text{d} g(1-\frac{1}{\mathsf{k}},-\mathsf{k}\lambda_i\alpha\tau^\star)/\text{d}\tau>0$ for  $i\in\{2,\cdots,\mathsf{j}-1\}$, we have $g(1-\frac{1}{\mathsf{k}},-\mathsf{k}\lambda_{\mathsf{j}-1}\alpha\tau^\star)\geq g(1-\frac{1}{\mathsf{k}},-\mathsf{k}\lambda_{\mathsf{j}-2}\alpha\tau^\star)\geq \cdots\geq g(1-\frac{1}{\mathsf{k}},-\mathsf{k}\lambda_2\alpha\tau^\star)$. As a result, it follows from~\eqref{eq::rate_consensus_d} that at $\tau=\tau^\star$ we have $\min\{\rho_{\tau,i}\}_{i=2}^{\mathsf{j}-1}=\rho_{\tau,2}$, which given~\eqref{eq::temp1} completes our proof.
\end{proof}

 % \solmaz{how does $\tau^\star_N$ and $\tau^\star_2$ change with $\kappa$? It looks like $W_0(\frac{1}{\kappa\ee})$ is a decreasing function of $\kappa$. what can we say using the fact that lower and upper bound on $\tau^\star$ are decreasing with $\kappa$?}\hossein{Yes, for $\kappa>0$ we have $W_0(\frac{1}{\kappa\ee})$ is a decreasing function of $\kappa$.}
Theorem~\ref{thm::consensus_prob} indicates that for any $\mathsf{k}>0$ there always exists a range of delay in $(0,\bar{\tau}]$ for which faster response can be achieved for the modified average consensus algorithm~\eqref{eq::consensus} relative to the original one~\eqref{eq::consensus-orig}. Next, our goal is to identify values of  $\mathsf{k}\in\realpositive$  for which the maximum driving effort $\vect{u}_{\tau,\mathsf{k}}(t)$  does not exceed the one for the original algorithm~\eqref{eq::consensus-orig} (for zero-input dynamics). However, before that we make the following statement about the maximum attainable rate by using outdated feedback.%\margin{the lemma on maximum attainable rate can go here}

\begin{lem}[Ultimate bound on the maximum attainable increase in the rate of convergence of~\eqref{eq::consensus}]\label{lem::ultimate_bound} {\rm For any  $\mathsf{k}\in\real_{\geq0}$, the ultimate bound on the maximum attainable rate of convergence for~\eqref{eq::consensus} by using outdated feedback is equal to
$
 (1-\mathsf{k})(1+\frac{1}{W_0(\frac{1-\mathsf{k}}{\mathsf{k}\ee})})\rho_0
$.}
\end{lem}
\begin{proof}
It follows from part~(f) of Lemma~\ref{thm::g(k,x)_prop} %part~(f) with $\gamma=\frac{\mathsf{k}-1}{\mathsf{k}}$ implies 
that $g(1-\frac{1}{\mathsf{k}},-\mathsf{k}\lambda_i\alpha\tau_i^\star)=\frac{1-\mathsf{k}}{\mathsf{k}W_0(\frac{1-\mathsf{k}}{\mathsf{k}\ee})}$ for any $i\in\{2,\cdots,N\}$. Then, given~\eqref{eq::rate_consensus_d} we have $\rho_{\tau}\leq\rho_{\tau,2}\leq\rho_{\tau^\star,2}=\big(\mathsf{k}g(1-\frac{1}{\mathsf{k}},-\mathsf{k}\lambda_2\alpha\tau_2^\star)+(1-\mathsf{k})\big)\alpha\lambda_2=(1-\mathsf{k})(1+\frac{1}{W_0(\frac{1-\mathsf{k}}{\mathsf{k}\ee})})\rho_0$, which concludes our proof.
\end{proof}

Next, we study how the maximum control effort of the agents while implementing for the modified algorithm~\eqref{eq::consensus} compares to that of the original average consensus algorithm~\eqref{eq::consensus-orig} any $\mathsf{k}\in\realpositive$.
The theorem  below indicates that for any $\mathsf{k}\in(0,1]$ using the outdated feedback does not increase the  maximum control effort while for $\mathsf{k}>1$ the maximum control effort is greater than the one of the original algorithm~\eqref{eq::consensus-orig}.

\begin{thm}[The maximum control effort for steering the zero-input dynamics of the algorithm~\eqref{eq::consensus}]\label{thm::max_con_eff} 
{\rm For a given $\alpha\in\real_{>0}$, let $\vect{u}_{0,0}$, and  $\vect{u}_{\tau,\mathsf{k}}(t)$ be respectively the network aggregated control input of the zero-input dynamics of~\eqref{eq::consensus-orig}, and~\eqref{eq::consensus} for any $\mathsf{k}\in\realpositive$ and $\tau\in\realpositive$. Then, for any $\tau\in[0,\bar{\tau}]$, where admissible delay bound $\bar{\tau}$ is given in Lemma~\ref{lem::admis_delay_modified_con}, the following assertions hold for $t\in\real_{\geq0}$:
\begin{itemize}
 \item[ (a) ] For $\mathsf{k}\in(0,1]$
we have $|\vect{u}_{\tau,\mathsf{k}}(t)|_{\infty}\leq|\vect{u}_{0,0}(t)|_{\infty}$.
\item [ (b) ] For $\mathsf{k}>1$ we have $|\vect{u}_{\tau,\mathsf{k}}(t)|_{\infty}\geq\ee^{(k-1)\alpha\lambda_2\tau}|\vect{u}_{0,0}(t)|_{\infty}$.
\end{itemize}}
\end{thm}

\begin{proof}
Consider the zero-input dynamics of~\eqref{eq::LaclacianEquivalent}, the equivalent representation of algorithm~\eqref{eq::consensus}. For the maximum control effort of algorithm~\eqref{eq::consensus} we have 
%\begin{align}\label{eq::max_controleffort}
 %   |\vect{u}_{\tau,\mathsf{k}}|_{\infty}=|-\alpha\,(1-\mathsf{k})\,\vect{L}\,\vect{x}(t)-\alpha\,\mathsf{k}\,\vect{L}\,\vect{x}(t-\tau)|_{\infty}.
%\end{align}
%Using the transformation matrix~\eqref{eq::T}, it can be written as
\begin{align}\label{eq::control_bound}
|\vect{u}_{\tau,k}(t)|_{\infty}=|-\alpha\,&(1-\mathsf{k})\,\vect{\Lambda}\,\vect{z}(t)-\alpha\,\mathsf{k}\,\vect{\Lambda}\,\vect{z}(t-\tau)|_{\infty}\nonumber\\=\alpha\max\{|(1-&\mathsf{k})\,\lambda_i\,z_i(t)+\,\mathsf{k}\,\lambda_i\,z_i(t-\tau)|_\infty\}_{i=2}^N.
\end{align}
Here we used the fact that $z_1(t)=0$.
Also, recalling~\eqref{eq::laclacian_equivalent}, for $\tau=0$ and any $i\in\{2,\cdots,N\}$ we have  $z_i(t)=\ee^{-\lambda_it}z_i(0)$, which gives $|\vect{u}_{0,0}(t)|_\infty=|\vect{u}_{0,0}(0)|_\infty=\alpha\max\{|\lambda_iz_i(0)|\}_{i=2}^N$. 

Next, we show that for  any $\tau\in(0,\bar{\tau})$ and $\mathsf{k}\in(0,1]$ $|\vect{u}_{\tau,\mathsf{k}}(t)|_\infty\leq\alpha\{|\lambda_iz_i(0)|\}_{i=2}^N$.  Notice that from~\eqref{eq::control_bound} we have  $|\vect{u}_{\tau,\mathsf{k}}(t)|_\infty\leq\alpha(1-\mathsf{k})\max\{|\lambda_iz_i(t)|_\infty\}+\alpha\mathsf{k}\max\{|\lambda_iz(t-\tau)|_\infty\}$. Also, recall that for $t\in[0,\tau)$ we have $z_i(t-\tau)=0$. Thus, to validate the statement~(a) it suffices to show that $|z_i(t)|_\infty=|z_i(0)|$. 
To this aim, consider  the trajectories $t\rightarrow z_{2:N}$ of~\eqref{eq::laclacian_equivalent}. Since set of dynamics~\eqref{eq::laclacian_equivalent} are exponentially stable with $-\alpha(1-\mathsf{k})\lambda_i\leq0$ and $-\alpha\mathsf{k}\lambda_i\leq0$, recalling Lemma~\ref{lem::max_x(t)} for any delay in the admissible range we have
$|z_i(t)|_\infty=\max_{s\in[-\tau,2\tau]}{|z_i(s)|}$ for any  $i\in\{2,\cdots,N\}$. Also, note that from~\eqref{eq::laclacian_equivalent}  we get
\begin{subequations}\label{eq::solu}
    \begin{align}
        z_i(t)&=0,\quad\quad\quad\quad\quad\quad\quad  ~\,t\in[-\tau,0),\label{eq::solu_a}\\
        z_i(t)&=\ee^{-\alpha(1-\mathsf{k})\lambda_i t}z_i(0),\quad t\in[0,\tau),\label{eq::solu_b}\\
        z_i(t)&= \ee^{-\alpha(1-\mathsf{k})\lambda_it}\!z_i(0)(1+\frac{\mathsf{k}}{(1-\mathsf{k})}(\ee^{-\alpha(1-\mathsf{k})\lambda_i(t-\tau)}\!-1))\!\!\!\!\!\nonumber\\&\quad t\in[\tau,2\tau]\label{eq::solu_c},
    \end{align}
\end{subequations}
which results in $\max_{s\in[-\tau,2\tau]}{|z_i(s)|}=|z_i(0)|$, and consequently $|z_i(t)|_\infty=|z_i(0)|$, which concludes statement~(a).
To validate part~(b) we proceed as follows. Recalling~\eqref{eq::control_bound} for $\mathsf{k}>1$ we have $|\vect{u}_{\tau,\mathsf{k}}(2\tau)|_\infty=\alpha\max\{|\mathsf{k}\lambda_iz_i(\tau)-(\mathsf{k}-1)\lambda_iz_i(2\tau)|\}_{i=2}^N\geq\alpha\max\{\mathsf{k}\lambda_i|z_i(\tau)|-(\mathsf{k}-1)\lambda_i|z_i(2\tau)|\}_{i=2}^N$. Also, from~\eqref{eq::solu_c} for $t\in[\tau,2\tau)$ we have $|z_i(2\tau)|\leq|z_i(\tau)|$, which gives $|\vect{u}_{\tau,\mathsf{k}}(2\tau)|_\infty=\alpha\max\{\mathsf{k}\lambda_i|z_i(\tau)|-(\mathsf{k}-1)\lambda_i|z_i(\tau)|\}_{i=2}^N=\alpha(2\mathsf{k}-1)\max\{\lambda_i|z_i(\tau)|\}_{i=2}^N$. Moreover, ~\eqref{eq::solu_b} implies that $z_i(\tau)=\ee^{\alpha(\mathsf{k}-1)\lambda_i \tau}z_i(0)$, which deduces $|\vect{u}_{\tau,\mathsf{k}}(2\tau)|\geq\alpha(2\mathsf{k}-1)\ee^{\alpha(\mathsf{k}-1)\lambda_i \tau}\max\{|\lambda_iz_i(0)|\}_{i=2}^N=(2\mathsf{k}-1)\ee^{\alpha(\mathsf{k}-1)\lambda_2 \tau}|\vect{u}_{0,0}|_\infty$.  Knowing $2\mathsf{k}-1\geq1$ and $|\vect{u}_{\tau,\mathsf{k}}(t)|_\infty\geq|\vect{u}_{\tau,\mathsf{k}}(2\tau)|_\infty$ we can conclude the proof.
\end{proof}

We close this section by a remark on how the split factor can be chosen based on the  expectations on the convergence rate, robustness to delay and managing the control effort.

\begin{rem}[Selecting $\mathsf{k}$ in the algorithm~\eqref{eq::consensus}]{\rm 
 Lemma~\ref{lem::admis_delay_modified_con}, Theorem~\ref{thm::consensus_prob} and Theorem~\ref{thm::max_con_eff} give insights on how we can choose the slit factor $\mathsf{k}\in\real$ given expectations on the algorithm's  acceleration, robustness to delay and control effort. Theorem~\ref{thm::max_con_eff} certifies that for any $\mathsf{k}\in(0,1]$, the rate of convergence we observe for any $\tau\in[0,\bar{\tau}]$ is attained without imposing any extra control effort on the agents. Therefore, assuming that the acceleration is expected without increasing the control effort, the split factor should be selected to satisfy $\mathsf{k}\in(0,1]$.
 According to Lemma~\ref{lem::ultimate_bound} the maximum attainable rate of convergence is an increasing function of $\mathsf{k}\in(0,1]$. Moreover, as $\mathsf{k}\rightarrow1$ 
 the ultimate bound on the rate of convergence converges to $\ee\rho_0$, which recovers the same bound established in~\cite[Theorem~4.4]{HM-SSK:18b}. On the other hand, as expected, as $\mathsf{k}\rightarrow0$ 
 the ultimate bound on the rate of convergence converges to  $\rho_0$. 
 %\NeedATT{ We note that for the special case of $\mathsf{k}=1$  this result is in accordance with~\cite[Theorem~4.4]{HM-SSK:18b} which  states that the maximum achievable rate due to delay is $\ee\rho_0$.}\hossein{To me, this part is extra}
 %Therefore, if the goal is acceleration without extra control effort, then we should use $\mathsf{k}=0$. 
 Finally, we observe from Lemma~\ref{lem::admis_delay_modified_con} that for $\mathsf{k}>0.5$ the admissible delay bound is finite, and thus the robustness of the algorithm to delay is not strong. In Contrary, the algorithm is robust with respect to any perturbation in delay for $\mathsf{k}\in(0,0.5]$%\margin{Hossein, please double check this remark}
 , because the admissible delay range for such split factors is $\real_{\geq0}$. Taking these observations into account, there exists a trade-off between robustness to delay and achieving higher acceleration when comes to choosing the split factor; $\mathsf{k}=1$ gives the maximum rate of convergence with the corresponding optimum delay while $\mathsf{k}=0.5$ results in robustness as well as higher rate of convergence relative to the original system~\eqref{eq::consensus-orig}.} %\hossein{For the value of $\tau^\star$, we cannot recover the exact value we found for $k=1$ in the other paper.}\solmaz{why is that? theoretically the result in the CDC paper for $\mathsf{k}=1$ is a special case of the work in this paper.}\hossein{What I meant is that over there we had $\tau^\star_N\leq\tau^\star\leq\tau^\star_2$. We also had an expression for the exact value of $\tau^\star$. Here, we are saying that $\tau^\star_N\leq\tau^\star\leq\tau^\star_2$ but we can not give any comment about exact value of $\tau^\star$.}
\end{rem}

\begin{figure}[t]
  \unitlength=0.5in \centering
  %%%%%%%%%%%%%%%%%%%%%% 1 %%%%%%%%%%%%%%%%%%%%%%%%%%%%%%%%%%%
     \subfloat{

    \centering
   \begin{tikzpicture}[auto,thick,scale=0.8, every node/.style={scale=0.8}]
   %\node (1) at (-2.3,0) {}; 
     \node (1) at (0,0) [ draw, minimum size=7pt,color=blue, circle, very thick, label=above:{}] {{\scriptsize 1}};
      \node (2) at (-1.3,-0) [ draw, minimum size=7pt,color=blue, circle,very thick, label=above:{}] {{\scriptsize 2}};
       \node (5) at (0,1.3) [ draw, minimum size=7pt,color=blue, circle, very thick, label=below:{}] {{\scriptsize 5}};
    \node (3) at (0,-1.3) [ draw, minimum size=7pt,color=blue, circle, very thick, label=above:{}] {{\scriptsize 3}};
    \node (4) at (1.3,0) [ draw, minimum size=7pt,color=blue, circle, very thick, label=below:{}] {{\scriptsize 4}};
    \node (L) at (5.5,0) [  minimum size=7pt,color=black] {{\small $\vect{L}\!=\!\begin{bmatrix}4&-1&-1&-1&-1\\-1&3&-1&0&-1\\-1&-1&3&-1&0\\-1&0&-1&3&-1\\-1&-1&0&-1&3\end{bmatrix}$}};
    \draw [<->] (1) to [out=-180,in=0] node[ fill=none,above]  {\small $1$} (2);
    \draw [<->] (1) to [out=0,in=180] node[ fill=none,above]  {\small $1$} (4);
      \draw [<->] (1) to [out=-90,in=90] node[ fill=none,right]  {\small $1$} (3);
   \draw [<->] (2) to [out=-90,in=180] node[ fill=none,below]  {\small $1$} (3);
    \draw [<->] (3) to  [out=0,in=-90] node[ fill=none,below]  {\small $1$} (4);
     \draw [<->] (5) to  node[ fill=none]  {\small $1$} (1);
      \draw [<->] (5) to [out=-180,in=90] node[ fill=none,above]  {\small $1$} (2);
      % \draw [<->] (2) to [out=-45,in=135] node[2, fill=none]  {\small $\!\!\!1$} (3);
       %\draw [->] (3) to [out=0,in=180] node[2, fill=none,below]  {\small $1$} (4);
        \draw [<->] (4) to [out=90,in=0] node[ fill=none,right]  {\small $1$} (5);
          %  \draw [->] (2) to [out=225,in=25] node[2, fill=none,above]  {\small $1$} (4);
  %  \draw [->] (2) to [out=-60,in=-120] node[2, fill=none,below]  {\small $0.5$} (4);
\end{tikzpicture}}
        \caption{{\small A connected graph of $5$ nodes.}}
    \label{fig:network}
\end{figure}
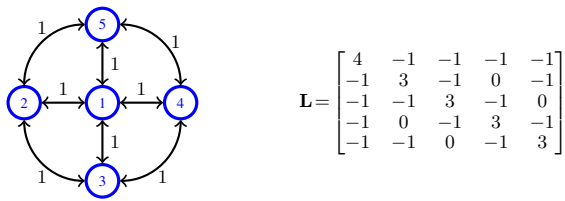

\vspace{-0.2in}
\section{Numerical Example}\label{sec::Num_ex}
We consider the modified average consensus algorithm~\eqref{eq::consensus} over the graph depicted in Fig.~\ref{fig:network}. The reference input of each agent $i\until{5}$ is chosen according to the 
first numerical example in~\cite{SSK-BVS-JC-RAF-KML-SM} to be the zero-order hold sampled points from the signal
 $\mathsf{r}^i(t)=a^i(2+\sin(\omega(t) t+\phi(t))+b^i$.
The idea discussed in~\cite{SSK-BVS-JC-RAF-KML-SM} is that the sensor agents sample the signal and should obtain the average of these sampled points before the next sampling arrives.  The parameters $a^i$ (the multiplicative sampling error) and $b^i$ (additive bias), $i\until{5}$, are chosen as the $i$th element of $[1.1,1,0.9,1.05,0.96]$ and $[-0.55,1,0.6,-0.9,-0.6]$, respectively. At each sampling time $\omega$ and $\phi$ are chosen randomly according to $N(0,0.25)$ and $N(0,(\pi/2)^2)$, where $N(\mu,\sigma)$ indicates the Gaussian distribution with mean $\mu$ and variance $\sigma$. We set the sampling rate at $2$Hz. This numerical example can be viewed as a simple abstraction for decentralized operations such as distributed sensor fusion where a dynamic or static average consensus algorithm is used to create the additive fusion terms in a distributed manner, e.g.,~\cite{ROS-JSS:05,ATK-JAF-AKR:13}. Since the  convergence of the average consensus algorithm is asymptotic, there is always an error when the algorithm is terminated in the finite inter-sampling time. Faster convergence is desired to reduce the residual error. 
 %This numerical example can be viewed as an abstract representation of a kind of response one abstraction of class of distributed solutions in for example data fusion which use a dynamic or static average consensus algorithm to create their additive fusion terms. It is desired that this algorithms converge fast to reduce (given the asymptotic convergence there is always an error when the consensus algorithms are terminated to complete the data fusion before the next sample arrives). with  or problems that use

For this example, in what follows, we study the response of the modified average consensus algorithm~\eqref{eq::consensus} for  
$\mathsf{k}\in\{-0.5,0,0.5,1,1.5\}$.  We note that the case of $\mathsf{k}=0$ gives the original (delay free) dynamic average consensus algorithm~\eqref{eq::consensus-orig} and thus is the baseline case that the rest of the cases should be compared to. For $\mathsf{k}\in\{-0.5,0,0.5,1,1.5\}$, the critical delay value $\bar{\tau}$ of the admissible delay range $(0,\bar{\tau})$ of~\eqref{eq::consensus}, respectively, is $\{\infty,\infty,\infty,0.32,0.18\}$~seconds. 
Figure~\ref{fig:res_log_ex2} illustrates how $\rho_\tau$ changes with $\tau$.  First, we note that for $\mathsf{k}=-0.5$ the rate of convergence decreases with delay. However, for positive values of $\mathsf{k}$ there is a range $(0,\tilde{\tau})$ for which $\rho_\tau>\rho_0$. For positive values of $\mathsf{k}$ we also observe monotonic increase until reaching $\tau^\star$ and then the monotonic decrease afterwards. The trend observed is in accordance with the results of Theorem~\ref{thm::consensus_prob}.%\margin{Hossein, please change the order of appearance of the figures according to the text}
We also can observe that as the $\mathsf{k}$ increases the maximum achievable rate of convergence increases also. %As illustrated in the figure, for $\mathsf{k=1.5}$ we have $\tilde{\tau}=0.14$ and the maximum rate of convergence can be achieved as $\rho_\tau^\star\approx 2\rho_0$ at $\tau^\star=0.11$.
Figure~\ref{fig:dynamic_consen1} shows the tracking response of agent~$2$ for $\mathsf{k}\in\{-0.5,0,0.5,1,1.5\}$ when the delay is $\tau=0.1$ (similar trend is observed for the other agents). As seen, the convergence rate of~\eqref{eq::consensus} is different for each value of $\mathsf{k}$. The fastest response is observed for $\mathsf{k}=1.5$ while $\mathsf{k}=-0.5$ shows the lowest one. The decrease of rate of convergence for $\mathsf{k}=-0.5$ and its increase for the positive values of $\mathsf{k}$ is in accordance with the trend certified by Theorem~\ref{thm::consensus_prob} (note that as seen in Figure~\ref{fig:res_log_ex2}, $\tau=0.1$ is in the rate increasing delay range of $(0,\tilde{\tau}$ of the cases corresponding to $\mathsf{k}\in\{0.5,1,1.5\}$). The desired effect of fast convergence shows itself in the smaller tracking error that is observed at the end of each sampling time, e.g., the tracking error in the first epoch for  $\mathsf{k}\in\{0.5,1,1.5\}$ is,  respectfully, \%13,~\%9, and \%0.5 that is an improvement over \%15 that corresponds to $\mathsf{k=0}$ (case of original algorithm). We note here that as can be seen in Fig.~\ref{fig:res_log_ex2}, $\tau=0.1$ is close to $\tau^\star$ of the case corresponding to $\mathsf{k}=1.5$. The same level of fast convergence can be achieved for the cases of $\mathsf{k}=1$ and $\mathsf{k}=0.5$ if one uses $\tau^\star$ corresponding to these split factors.

Figure~\ref{fig:con_eff} shows the maximum control effort of zero-input dynamics of the algorithm~\eqref{eq::consensus} over time corresponding to $\tau=0.1$ and different values of $\mathsf{k}\in\{0,0.5,1,1.5\}$. For $\mathsf{k}=1.5$ the maximum control effort exceeds the value for the original consensus algorithm (case of $\mathsf{k=0}$). But, for $\mathsf{k}=1$ and $\mathsf{k}=0.5$ the maximum control effort is equal or less than the case $\mathsf{k}=0$. The trend observed above is in accordance with Theorem~\ref{thm::max_con_eff}.

\begin{figure}[t]
    \centering
        \includegraphics[trim=0pt 0 0 0,clip,scale=0.75]{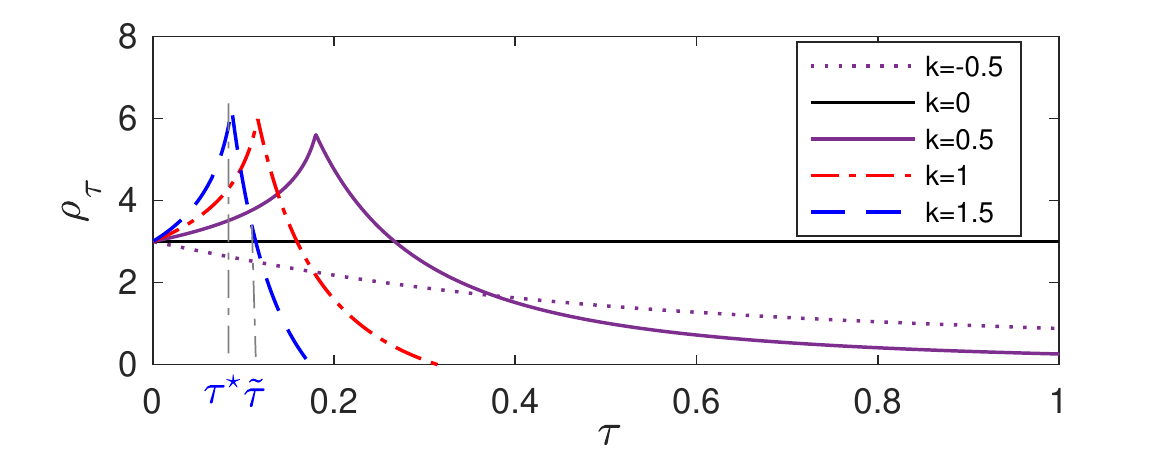}
   \vspace{-0.2in}
   \caption{{\small The rate of convergence $\rho_\tau$ of the modified average consensus algorithm~\eqref{eq::consensus}  over the graph in Fig.~\ref{fig:network} for different values of feedback gain $\mathsf{k}\in\{-0.5,0,0.5,1,1.5\}$. 
 For the example case of $\mathsf{k=1.5}$, note that $\tilde{\tau}=0.14$ and the maximum rate of convergence that can be achieved is $\rho_\tau^\star\approx 2\rho_0$ at $\tau^\star=0.11$.}
 }
    \label{fig:res_log_ex2}
\end{figure}
%\vspace{-0.1in}
\begin{figure}[t]
    \centering
   % \includegraphics[scale=0.65]{Fig/8}
      %\subfloat[]{
    %  \includegraphics[trim=5pt 0 0 0,clip,scale=0.7]{Fig/delayjournal1.eps }
 %   \includegraphics[trim=5pt 0 0 0,clip,scale=0.65]{Fig/delayjournal1_1.eps}
     \includegraphics[trim=5pt 0 0 0,clip,scale=0.7]{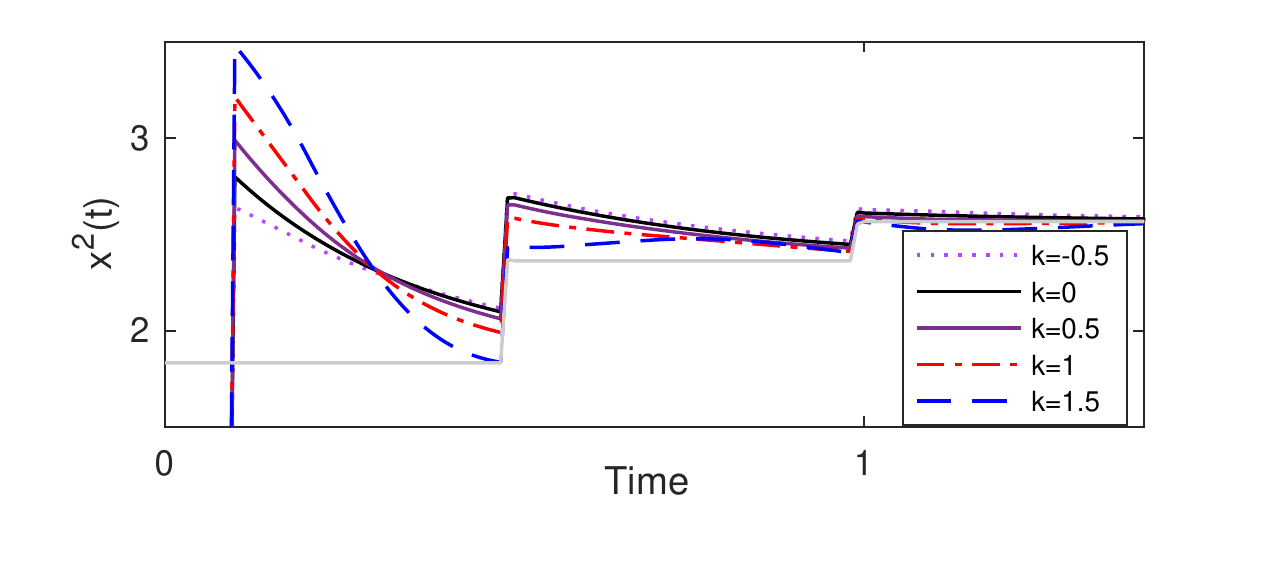}
    \vspace{-0.2in}
   \caption{{\small The trajectory of local state of agent~$2$ executed by the algorithm~\eqref{eq::consensus}  over the graph in Fig.~\ref{fig:network} for $\tau=0.1$ and  different different values of feedback gain $\mathsf{k}\in\{-0.5,0,0.5,1,1.5\}$.}
 }
    \label{fig:dynamic_consen1}
\end{figure}

\begin{figure}[t]
    \centering
   % \includegraphics[scale=0.65]{Fig/8}
      %\subfloat[]{
      \includegraphics[trim=0pt 0 0 0,clip,scale=0.75]{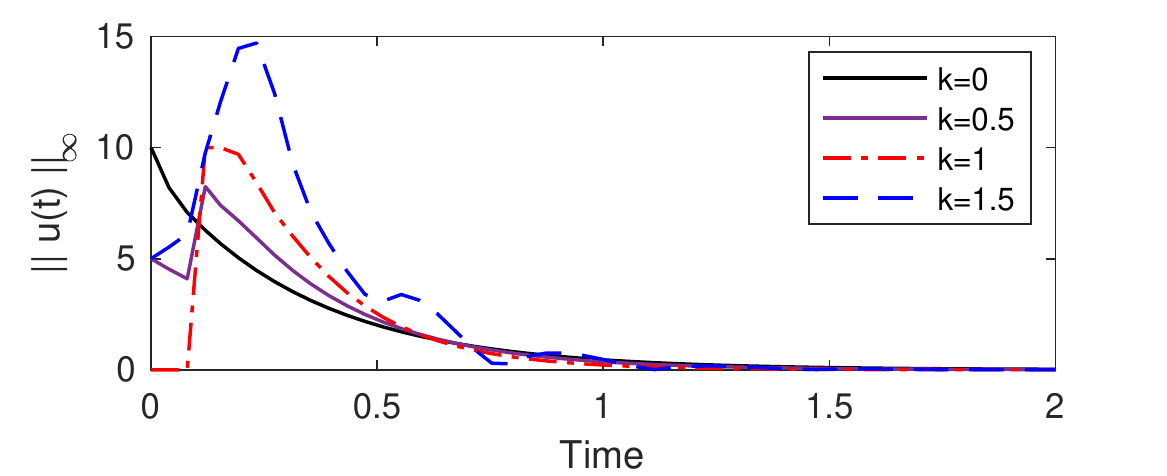}
    \vspace{-0.2in}
   \caption{{\small The maximum control effort executed by the algorithm~\eqref{eq::consensus}  over the graph in Fig.~  \ref{fig:network} for $\tau=0.1$ and different values of feedback gain $\mathsf{k}\in\{0,0.5,1,1.5\}$.}
 }
    \label{fig:con_eff}
\end{figure}

%\begin{figure}[t]
 %   \centering
   % \includegraphics[scale=0.65]{Fig/8}
  %      \includegraphics[trim=1pt 0 0 0,clip,scale=0.35]{Fig/rho_ex5}
   % \includegraphics[scale=0.85]{Fig/rho}
 %  \caption{
  % } 
  %  \label{fig:res_log_ex1}
%\end{figure}

%\margin{can you put the figures 3 and 4 side by side in one figure }

\section{Conclusion}\label{sec::conclu}
%\solmaz{we need a numerical example with more drastic improvement}
We  analyzed  the effect of using an affine combination of immediate and outdated disagreement feedbacks in increasing the rate of convergence of a dynamic average consensus algorithm. The modified algorithm  has the same ultimate tracking accuracy but with the right choices of the delay and the affine combination factor, can have faster convergence. Our study produced a set of closed-form expressions to specify the admissible delay range, the delay range for which the system experiences increase in its rate of convergence and a range that the optimum time delay corresponding to the maximum rate of convergence lies.
We also examined the range of affine combination factor for which the outdated feedback can be used to improve the convergence of the algorithm without increasing the control effort. To develop our results we used the Lambert W function to obtain the rate of convergence of our algorithm under study in the presence of the delay. Our future work includes extending our results for dynamic consensus algorithms over directed graphs and also investigating the use of outdated feedback in increasing the rate of convergence of other distributed algorithms for networked systems such as leader-follower algorithms.

%\section*{Acknowledgement}
%This work is supported by NSF CAREER award ECCS-1653838.
%\section*{References}
%\bibliographystyle{ieeetr}%
%\bibliography{bib/alias,bib/Reference}

%\solmaz{revise [27] and [32], they have now appeared in paper issue of their respective journal and have issue number}
\bibliographystyle{ieeetr}%
\bibliography{bib/alias,bib/Reference}

\begin{appendices}
\section{Delay gain function}\label{sec::Prelim}

The  lemma below highlights some of the  properties of the delay gain function $g(\gamma,\mathsf{x})$. Figure~\ref{Fig::g(k,x)} gives some graphical representation for the properties discussed in this lemma.

\begin{lem}[Properties of $g(\gamma,\mathsf{x})$] \label{thm::g(k,x)_prop}
{\rm The following assertions hold for the delay gain function~\eqref{eq::gain_rate} with $\gamma,\mathsf{x}\in\real$ :
\begin{itemize}
    \item[(a)] For any $\gamma\in\real$ we have $\lim_{\mathsf{x}\rightarrow0}g(\gamma,\mathsf{x})=1$.
    \item[(b)] For any $\gamma>1$ and $\mathsf{x}\in\real_{>0}$ we have $g(\gamma,\mathsf{x})<\gamma$.
    %For any $k\in\real_{>1}$ and $\mathsf{x}\in\real_{>0}$ we have $g(k,\mathsf{x})<k$.
    \item[(c)] For any $\gamma>1$ and $\mathsf{x}\in\real_{>0}$, $g(\gamma,\mathsf{x})$ is a strictly increasing function of $\mathsf{x}$.
    %For any $k\in\real_{>1}$ and $\mathsf{x}\in\real_{>0}$, $g(k,\mathsf{x})$ is a strictly increasing function of $\mathsf{x}$.
    \item[(d)] Let  $\mathsf{x}\in(\bar{\mathsf{x}},0)$, where $\bar{\mathsf{x}}=\arccos(\gamma)/\sqrt{1-\gamma^2}$. Then, for any $\gamma<1$ (respectively $\gamma>1$) we have $g(\gamma,\mathsf{x})>\gamma$ (respectively $g(\gamma,\mathsf{x})<\gamma$).
    \item[(e)] For any $\gamma<1$ and $\mathsf{x}\in\real_{<0}$, $g(\gamma,\mathsf{x})$ is a strictly decreasing function of $\mathsf{x}$ for any $\mathsf{x}\in[\mathsf{x}^\star,0)\subset(\bar{\mathsf{x}},0)$, and a strictly increasing function of $\mathsf{x}$ for any $\mathsf{x}<\mathsf{x}^\star$, where $\mathsf{x}^\star=\frac{1}{\gamma}W_0(-\frac{\gamma}{\ee})$ when $\gamma\neq0$ and $\mathsf{x}^\star=-\frac{1}{\ee}$%\hossein{I should be negative number} \solmaz{that is why I highlighted it. The limit does not seem to give a negative number}
    when $\gamma=0$. 
    \item[(f)] For any $\gamma<1$ and $\mathsf{x}\in\real_{<0}$, the maximum value of $g(\gamma,\mathsf{x})$ occurs at $\mathsf{x}^{\star}=\frac{1}{\gamma}W_0(-\frac{\gamma}{\ee})$ 
    where $g(\gamma,\mathsf{x}^{\star})=\frac{-\gamma}{W_0(-\frac{\gamma}{\ee})}$ when $\gamma\neq0$, and at $\mathsf{x}^\star=-\frac{1}{\ee}$ where $g(\gamma,\mathsf{x}^{\star})=\ee$ when $\gamma=0$.
    \item[(g)] For any $\gamma<1$ and $\mathsf{x}\in\real_{<0}$, $g(\gamma,\mathsf{x})>1$ if and only if $\mathsf{x}\in(\tilde{\mathsf{x}},0)$ where $\tilde{\mathsf{x}}$ is the unique solution of $g(\gamma,\mathsf{x})=1$ in $(\bar{\mathsf{x}},0)$.
\end{itemize}}
\end{lem}

%\begin{comment}
\begin{figure}[t!]\label{Fig::g(k,x)}%[htbp]
  \unitlength=0.5in
   \captionsetup{skip=13pt}\centering
  %%%%%%%%%%%%%%%%%%%%%% 1 %%%%%%%%%%%%%%%%%%%%%%%%%%%%%%%%%%%
  \subfloat[$\gamma>1$]{
    \includegraphics[trim={5pt 2pt 10pt 0},clip,scale=0.75]{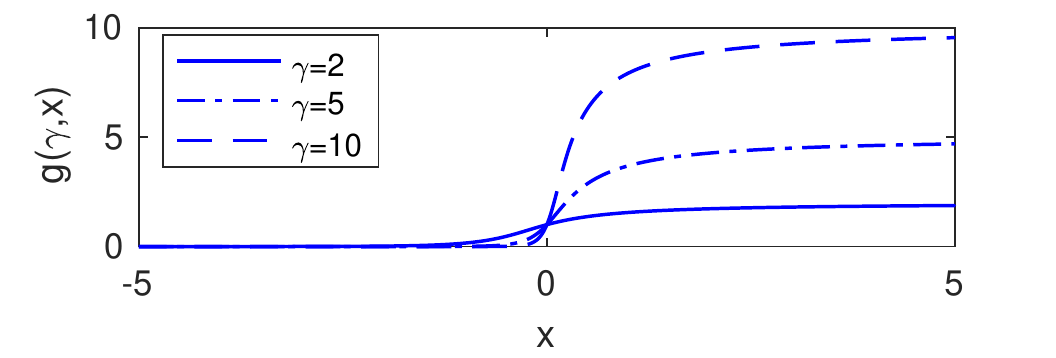} }
 \\\vspace{-0.15in}
  %%%%%%%%%%%%%%%%%%%%%% 4 %%%%%%%%%%%%%%%%%%%%%%%%%%%%%%%%%%%
 \subfloat[$0<\gamma<1$]{
    \includegraphics[trim={5pt 2pt 10pt 0},clip,scale=0.75]{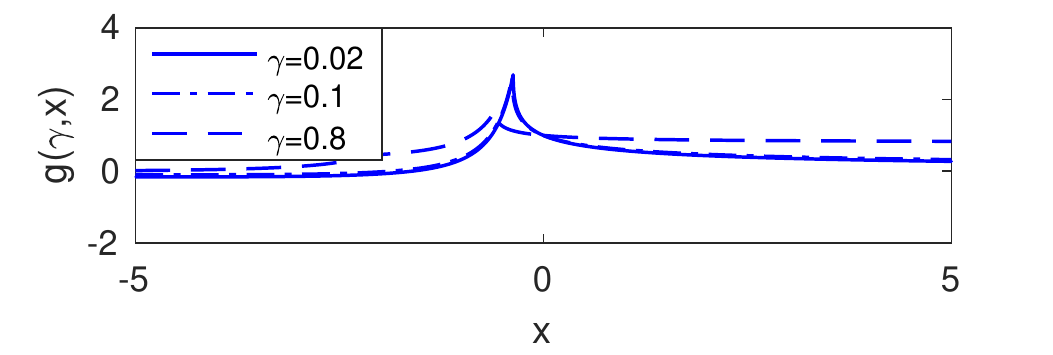} % }
  }\\ \vspace{-0.15in}
  %%%%%%%%%%%%%%%%%%%%%% 4 %%%%%%%%%%%%%%%%%%%%%%%%%%%%%%%%%%%
  \subfloat[$-1<\gamma<0$]{
    \includegraphics[trim={2pt 0pt 10pt 0},clip,scale=0.75]{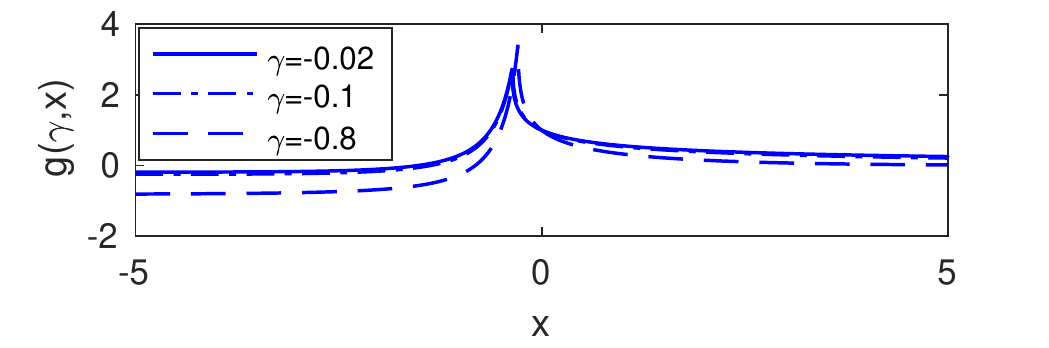} % }
  }\\\vspace{-0.10in}
     \subfloat[$\gamma<-1$]{
    \includegraphics[trim={2pt 0pt 10pt 0},clip,scale=0.75]{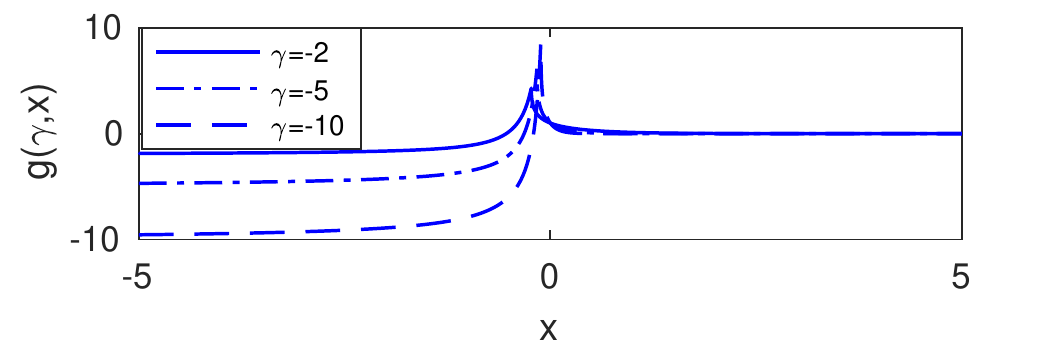} % }
  }\vspace{-0.10in}
  \caption{
 {\small The delay gain function for different values of $\mathsf{x},\gamma$.} 
  }\label{Fig::g(k,x)}
\end{figure}
%\end{comment}
The proof of this lemma invokes various properties of the Lambert W function listed in Section~\ref{sec::prelim} and is given in Appendix~\ref{sec::APP_B}. The next theorem, whose proof relies on the results of Lemma~\ref{thm::g(k,x)_prop}, and is also given in Appendix~\ref{sec::APP_B}, characterizes the effect of  delay on the rate of convergence of scalar time-delayed system~\eqref{DDE_scalar}.
The tightest estimate of the rate of convergence of~\eqref{DDE_scalar} is characterized by the magnitude of the real part of the rightmost root of its characteristic equation $s=\frac{1}{\tau}W_0(\alpha\,\tau\,\ee^{-\tau\mathsf{b}})+\mathsf{b}$ 
%\begin{align}s=\frac{1}{\tau}W_0(\alpha\,\tau\,\ee^{-\tau\mathsf{b}})+\mathsf{b}.\end{align}
(recall Lemma~  \ref{lem::LambertReal} and \eqref{eq::W_0_facts_a}).
That is (see~\cite[Corollary~1]{SD-JN-AGU:11})
\begin{align}\label{eq::rho}
\rho_\tau=\,\,-\frac{1}{\tau}\re{W_0(\mathsf{a}\,\tau\,\ee^{-\tau\mathsf{b}})}-\mathsf{b}.
\end{align}

Recalling~\eqref{eq::gain_rate},
%To compare the rate of convergence~\eqref{eq::rho} to the delay free rate, we  define the \emph{delay gain function} 
%\begin{align}\label{eq::gain_rate}
%\!\!\!g(\gamma,\mathsf{x})=\begin{cases}
%\frac{1}{\mathsf{x}}\re{W_0(\mathsf{x}\,\ee^{\gamma\,\mathsf{x}})},\quad\quad \mathsf{x}\in\real\backslash\{0\},\quad\quad \\ 1,  \quad\quad\quad\quad\quad\quad\quad\quad\quad\quad \mathsf{x}=0, \end{cases}
%\end{align}
%where $\mathsf{x},\gamma\in\real$. Then,
we write~\eqref{eq::rho} as
\begin{align}\label{rate_g(x)}
\rho_\tau=-(g(\gamma,\mathsf{x})\,\mathsf{a}+\mathsf{b}) = -(g(\gamma,\mathsf{x})-\gamma)\,\mathsf{a}
\end{align}
%\solmaz{why do we need to write rate in this form~$-(1-\frac{g(\gamma,\mathsf{x})}{\gamma})\,\mathsf{b}$?, this way, technically we need to require $\gamma\neq0$}
where $\mathsf{x}\!=\!\mathsf{a}\tau$ and $\gamma\!=\!-\frac{\mathsf{b}}{\mathsf{a}}$. It follows from~\eqref{eq::limWz-z}~that 
\begin{align}
\lim_{\tau\to0}g(\gamma,\mathsf{a}\tau)=1.
\end{align}
Therefore, as expected, $\lim_{\tau\to0}\rho_\tau=\rho_0$, where
\begin{align}
    \rho_0=-(\mathsf{a}+\mathsf{b})=-(1-\gamma)\mathsf{a}.
\end{align}

system~\eqref{DDE_scalar} in terms of different values of $\mathsf{a},\mathsf{b}\in\real$, $\mathsf{a}\neq0$ satisfying $\mathsf{a}+\mathsf{b}<0$.
\begin{thm}[Effect of delay on the rate of convergence of delayed system~\eqref{DDE_scalar}]\label{thm::delay_effect_scalar}
{\rm Consider system~\eqref{DDE_scalar} with $\mathsf{a}\in\real\backslash\{0\}$ and $\mathsf{b}\in\real$ such that $\mathsf{a}+\mathsf{b}<0$, whose rate of convergence $\rho_\tau$ is specified by~\eqref{rate_g(x)}. Consider also the delay gain function~\eqref{eq::gain_rate} with $\gamma=-\frac{\mathsf{b}}{\mathsf{a}}$ and $\mathsf{x}=\mathsf{a}\tau$. Then,

\begin{itemize}
    \item[(a)] for $\mathsf{a}>0$ and $\mathsf{b}<0$ the system~\eqref{DDE_scalar} is exponentially stable for any $\tau\in\real_{\geq0}$. Moreover, the rate of convergence decreases by increasing $\tau\in\real_{\geq0}$.
       \item[(b)] for $\mathsf{a}<0$ and $\mathsf{b}\in\real$, $\rho_\tau>\rho_0$ if and only if $\tau\in[0,\tilde{\tau})\subset[0,\bar{\tau})$ where $\tilde{\tau}$ is the unique solution of $g(\gamma,\mathsf{a}\tau)=1$ in $(0,\bar{\tau})$  and $\bar{\tau}$ is specified by
     \begin{align}\label{eq::admissible_range_tau}
  \bar{\tau}=\text{arccos}(-\mathsf{b}/\mathsf{a})/\sqrt{\mathsf{a}^2-\mathsf{b}^2}.
    \end{align}
      % ~\eqref{eq::admissible_range_tau}. 
       Moreover, $\rho_{\tau}$ is monotonically increasing (resp. decreasing) with $\tau$ for any $\tau\in[0,\tau^\star)\subset[0,\bar{\tau})$ (resp. $\tau\in(\tau^\star,\bar{\tau})\subset[0,\bar{\tau})$), where $\tau^\star=-\frac{1}{\mathsf{b}}W_0(\frac{\mathsf{b}}{\mathsf{a}\,\ee})$ when $\mathsf{b}\neq0$ and $\tau^\star=-\frac{1}{\mathsf{a}\,\ee}$ when $\mathsf{b}=0$. Finally, the maximum rate of convergence of
       $\rho_\tau^{\star}=-(1+\frac{1}{W_0(\frac{\mathsf{b}}{\mathsf{a}\ee})})\mathsf{b}$ when $\mathsf{b}\neq0$ and $\rho_\tau^{\star}=-\mathsf{a}\ee$ when $\mathsf{b}=0$ is obtained at $\tau=\tau^\star$.
\end{itemize}}
\end{thm}
In developing our results we also invoke the following result. 
\begin{lem}(maximum value of  the trajectory of~\eqref{DDE_scalar}
~\cite[Theorem~2.10]{AI-EL-ST:02})\label{lem::max_x(t)}
{\rm For the time delay system~\eqref{DDE_scalar} and any $\tau\in(0,\bar{\tau}]$ with $\mathsf{a},\mathsf{b}\in\real_{<0}$ the following holds 
\begin{align}\label{eq::max_x}
  |x(t)|_{\infty}=\text{max}_{s\in[-\tau,2\tau]}|x(s)|.
\end{align}
}
\end{lem}

%We close this section with a remark on use of  Theorem~   \ref{thm::delay_effect_scalar} to  study the effect of delay on the rate of convergence of the linear time-delayed system
%\begin{align}\label{eq::system_matrix}
%\dvect{x}(t)&=\vect{A}\,\vect{x}(t-\tau)+\vect{B}\,
%\vect{x}(t)&t\in\real_{\geq0},\\
%\vect{x}(t) &\in\real^n,   &  t\in [-\tau,0],\nonumber
%\end{align}
%where $\vect{A},\vect{B}\in\real^{n\times n}$ are diagonalizable matrices, $\vect{A}\vect{B}=\vect{B}\vect{A}$ and eigenvalues of $\vect{A}$ and $\vect{B}$ are real. For this system, because the system matrices commute, there exists a similarity transformation that simultaneously diagonalizes these two matrices~\cite{RAH-CRJ:90}. That is,  there exists a similarity transformation $\vect{z}=\vect{T}\vect{x}$ that represents the system~\eqref{eq::system_matrix} in the equivalent form
%\begin{align}\label{eq::system_matrix_equ}
%\dot{z}^i(t)=\alpha_iz^i(t-\tau)+\mathsf{b}_iz^i(t)   \quad\quad i\in\{1,\cdots,n\}
%\end{align}
%where $\{\alpha_i\}_{i=1}^n$ and $\{\mathsf{b}_i\}_{i=1}^n$ are the eigenvalues of $\vect{A}$ and $\vect{B}$, respectively. Given the equivalent decomposed representation~\eqref{eq::system_matrix_equ}, we can now use the results from Theorem~   \ref{thm::delay_effect_scalar} to study the effect of delay on the rate of convergence of~\eqref{eq::system_matrix}. The proceeding section demonstrates an example of such a study.

%%%%%%%%%%%%%%%%%%%%%%%%%%
%%%%%%%%%%%%%%%%%%%%%%%%%%%%%%

%\hline

\section{Proofs of Lemma~\ref{thm::g(k,x)_prop} and Theorem~  \ref{thm::delay_effect_scalar}}\label{sec::APP_B}

\begin{proof}[Proof of Lemma~\ref{thm::g(k,x)_prop}]
Part~(a) can be readily deduced by invoking~\eqref{eq::lambert_derivative} since $W_0(\mathsf{x}\,\ee^{\gamma\,\mathsf{x}})\rightarrow \mathsf{x}\,\ee^{\gamma\,\mathsf{x}}$ as $\mathsf{x}\rightarrow0$. To prove statement (b) we proceed as follows. Let $q=\mathsf{x}\,\ee^{\gamma\,\mathsf{x}}$. 
%\solmaz{Since you are using u below for imaginary part of Lambert value, here use a different notation, perhaps $q=\mathsf{x}\,\ee^{k\,\mathsf{x}}$???} 
Since $x\in\real_{>0}$, then $q\in\real_{>0}$. As a result, given the properties of Lambert W function reviewed in Section~  \ref{sec::prelim}, we can write $\mathsf{x}=\frac{1}{\gamma}W_0(\gamma q)$ and $\re{W_0(q)}=W_0(q)$,
%\st{(recall that $W_0(\mathsf{x})$ is the only real solution of Lambert W function for any $\mathsf{x}\in\real_{>0}$).}, 
which allows us to represent  $g(\gamma,\mathsf{x})$ as
\begin{align}\label{eq::k_prop_stat-b}
    g(\gamma,\mathsf{x})=\frac{W_0(q)}{W_0(\gamma\,q)}\,\gamma,\quad\quad  \text{for~~}x\in\real_{>0}.
\end{align}
Since for $\gamma>1$ we have $q<\gamma\,q$, by invoking  Lemma~  \ref{lem::W0-increase-xpos} we obtain $\frac{W_0(q)}{W_0(\gamma\,q)}<1$, which together with $W_0(q)\in\real_{>0}$ and $W_0(\gamma\,q)\in\real_{>0}$ validates statement (b) from~\eqref{eq::k_prop_stat-b}.

%\st{To validate statement~(b) since $k>1$ it suffices to show that $\re{W_0(x)}\leq\re{W_0(y)}$ holds for any $\{x,y\in\real_{>0}|y>x\}$. For any $\mathsf{x}\in\real_{>0}$ $\re{W_0(\mathsf{x})}$ is an increasing function since\\
%$\frac{\text{d}\,W(\mathsf{x})}{\text{d}\,\mathsf{x}}=\frac{1}{\mathsf{x}+\ee^{W(\mathsf{x})}}>0,$
%so it concludes our statement. }

Next, we validate statement (c). The derivative of $g(\gamma,\mathsf{x})$ with respect to $\mathsf{x}\in\real$ is 
\begin{align}\label{eq::g_derivative}
    &\frac{\text{d}\,g(\gamma,\mathsf{x})}{\text{d}\,\mathsf{x}}=\frac{(1+\gamma\,\mathsf{x})\ee^{\gamma\,\mathsf{x}}}{\mathsf{x}}\re{\frac{1}{\mathsf{x}\,\ee^{\gamma\,\mathsf{x}}+\ee^{W_0(\mathsf{x}\,\ee^{\gamma\,\mathsf{x}})}}}-\\&\frac{1}{\mathsf{x}^2}\re{W_0(\mathsf{x}\,\ee^{\gamma\,\mathsf{x}})}=\frac{(1+\gamma\,\mathsf{x})}{\mathsf{x}^2}\re{\frac{W_0(\mathsf{x}\,\ee^{\gamma\,\mathsf{x}})}{W_0(\mathsf{x}\,\ee^{\gamma\,\mathsf{x}})+1}}-\nonumber\\&\frac{1}{\mathsf{x}^2}\re{W_0(\mathsf{x}\,\ee^{\gamma\,\mathsf{x}})}\!=\!\frac{1}{\mathsf{x}^2}\re{\frac{(\gamma\,\mathsf{x}\!-\!W_0(\mathsf{x}\ee^{\gamma\,\mathsf{x}}))W_0(\mathsf{x}\ee^{\gamma\,\mathsf{x}})}{(W_0(\mathsf{x}\,\ee^{\gamma\,\mathsf{x}})+1)}},\nonumber
\end{align}
for $\mathsf{x}\,\ee^{\gamma\,\mathsf{x}}\neq-\frac{1}{\ee}$.
Recall~\eqref{eq::W_0_facts_c} that  $\re{W_0(\mathsf{z})}+1>0$ for any $\mathsf{z}\in\real\backslash\{-\frac{1}{\ee}\}$ and $\re{W_0(\mathsf{z})}=W_0(\mathsf{z})>0$ for any $\mathsf{z}\in\real_{>0}$. Note also that we have already shown that for any $\gamma>1$ and $\mathsf{x}>0$ we have  $g(\gamma,\mathsf{x})<\gamma$ which gives $\gamma\,\mathsf{x}-W_0(\mathsf{x}\,\ee^{\gamma\,\mathsf{x}})>0$. Therefore, for $\gamma>1$ and $\mathsf{x}\in\real_{>0}$ from~\eqref{eq::g_derivative} we obtain $\frac{\text{d}\,g(\gamma,\mathsf{x})}{\text{d}\,\mathsf{x}}>0$, which validates statement (c).
%Recall that  $\re{W_0(\mathsf{x})}+1>0$ for any $\mathsf{x}\in\real\solmaz{\backslash\{-\frac{1}{\ee}\}}$ and $\re{W_0(\mathsf{x})}=W_0(\mathsf{x})>0$ for any $\mathsf{x}\in\real_{>0}$. Also for any $k>1$ and $\mathsf{x}>0$, part~(b) implies that $g(k,\mathsf{x})<k$, which gives $k\,\mathsf{x}-W_0(\mathsf{x}\ee^{k\,\mathsf{x}})>0$. Therefore,~\eqref{eq::g_derivative} is positive for any $k>1$ and $\mathsf{x}>0$ which results in statement~(c).

To validate statement~(d), consider $x\in(\bar{x},0]$. For $x\rightarrow0^-$ we have $g(\gamma,x)\rightarrow1$. So, for $\gamma<1$ (respectively $\gamma>1$) we get $g(\gamma,x)>\gamma$ (respectively $g(\gamma,x)<\gamma$) as $x\rightarrow0^-$. Moreover, we know that the admissible bound, $x=\bar{x}$ is the first point that $g(\gamma,x)=\gamma$ holds. So, since $g(\gamma,x)$ is a continuous function, for any $x\in(\bar{x},0]$ we have $g(\gamma,x)>\gamma$ for $\gamma<1$, and $g(\gamma,x)<\gamma$ for $\gamma>1$.

%$$----------$$
For proof of statement~(e) we proceed as follows. 
Recall the properties of Lambert $W_0$ function in~\eqref{eq::W_0_facts}. Note that for $0<\gamma<1$, we have $-1<W_0(-\frac{\gamma}{\ee})<0$ and for $\gamma<0$, we have $W_0(-\frac{\gamma}{\ee})>0$. Also recall that $W_0(0)=0$. Therefore, for $\gamma<1$ and $\gamma\neq0$, we have $\frac{1}{\gamma}W_0(-\frac{\gamma}{\ee})\in\real_{<0}$. Now for $\gamma<1$  consider $\mathsf{x}\in[\frac{1}{\gamma}W_0(-\frac{\gamma}{\ee}),0)$ for $\gamma\neq0$ and $\mathsf{x}\in[-\frac{1}{\ee},0)$ for $\gamma=0$. For such $\mathsf{x}$, we have $\mathsf{x}\,\ee^{\gamma\,\mathsf{x}}\in\real_{<0}$. For $f(\mathsf{x})=\mathsf{x}\,\ee^{\gamma\mathsf{x}},$ with $\mathsf{x},\gamma\in\real$ we know $\frac{d\,f}{d\,\mathsf{x}}=(1+\gamma\,\mathsf{x})\ee^{\gamma\mathsf{x}}>0$ for any $\mathsf{x}\in(-\frac{1}{\ee},0]$ and $\gamma<1$, i.e., $f(\mathsf{x})$ is a strictly increasing continuous function.
Because the solutions of $\mathsf{z}\ee^{\gamma\mathsf{z}}=-\frac{1}{\ee}$ are $\mathsf{z}=\frac{1}{\gamma}W_{l}(-\frac{\gamma}{\ee})$, $l=\{-1,0\}$ for $\gamma\neq0$ and $\mathsf{z}=-\frac{1}{\ee}$ for $\gamma=0$, for $\mathsf{x}\in[\frac{1}{\gamma}W_0(-\frac{\gamma}{\ee}),0)$ we have $\mathsf{x}\ee^{\gamma\mathsf{x}}\in(-\frac{1}{\ee},0]$ and then $W_0(\mathsf{x}\ee^{\gamma\,\mathsf{x}})\in\real_{<0}$ (recall \eqref{eq::W_0_facts_a}). Next, note that by statement~(d) we have $\gamma\,\mathsf{x}-W_0(\mathsf{x}\,\ee^{\gamma\,\mathsf{x}})=\mathsf{x}(\gamma-g(\mathsf{x},\gamma))>0$ for $\mathsf{x}\in(\bar{\mathsf{x}},0]$. %\solmaz{admissible range (Q:what do you mean by admissible range)}\hossein{$\mathsf{x}\in(\bar{\mathsf{x}},0]$. Lets consider $x\in(\bar{x},0]$. For $x\rightarrow0^-$ we have $g(k,x)\rightarrow1$. So, for $k<1$ ($k>1$) we get $g(k,x)>k$ ($g(k,x)<k$). Moreover, we know that the admissible bound, $x=\bar{x}$ is the first point that $g(k,x)=k$ holds. So, since $g(k,x)$ is a continuous function, for any $x\in(\bar{x},0]$ we have $g(k,x)>k$ ($g(k,x)<k$).} 
Therefore $\frac{d\,g(\mathsf{x},\gamma)}{d\,\mathsf{x}}<0$ can be inferred from~\eqref{eq::g_derivative}. 
Next, for $\mathsf{x}<\frac{1}{\gamma}W_0(-\frac{\gamma}{\ee})$, let $W_0(\mathsf{x}\,\ee^{\gamma\,\mathsf{x}})=w+\mathsf{i}\,u$. %\solmaz{in previous theorem u was the imaginary part,it would be better here to use u as imaginary part of change the previous theorem} 
Then,~\eqref{eq::g_derivative} can be written as $\frac{\text{d}\,g(\gamma,\mathsf{x})}{\text{d}\,\mathsf{x}}=\frac{1}{\mathsf{x}^2}\re{\frac{(\gamma\,\mathsf{x}-(w+\mathsf{i}\,u))(w+\mathsf{i}\,u)}{((w+\mathsf{i}\,u)+1)}}=\frac{1}{\mathsf{x}^2((w+1)^2+u^2)}((\gamma\mathsf{x}-w)(w^2+u^2+w)+u^2).$
In addition, we have $w=-u\cot{u}$ since $\im{\mathsf{x}\,\ee^{\gamma\,\mathsf{x}}}=0$, which gives
 $\frac{\text{d}\,g(\gamma,\mathsf{x})}{\text{d}\,\mathsf{x}}=\frac{1}{\mathsf{x}^2u^2((\cot{u}+1)^2+1)}((\gamma\mathsf{x}+u\cot{u})
    (u^2\cot^2{u}+u^2-u\cot{u})+u^2)>0$.
%\begin{align*}
%    \frac{\text{d}\,g(\gamma,\mathsf{x})}{\text{d}\,\mathsf{x}}&=\frac{1}{\mathsf{x}^2u^2((\cot{u}+1)^2+1)}((\gamma\mathsf{x}+u\cot{u})\\&
%    (u^2\cot^2{u}+u^2-u\cot{u})+u^2)>0.\nonumber
%\end{align*}
%\begin{align*}
%    &\frac{\text{d}\,g(\gamma,\mathsf{x})}{\text{d}\,\mathsf{x}}=\frac{1}{\mathsf{x}^2}\re{\frac{(\gamma\,\mathsf{x}-(w+\mathsf{i}\,u))(w+\mathsf{i}\,u)}{((w+\mathsf{i}\,u)+1)}}\nonumber\\&=\frac{1}{\mathsf{x}^2((w+1)^2+u^2)}((\gamma\mathsf{x}-w)(w^2+u^2+w)+u^2).\nonumber\\&
%\end{align*}
%In addition, we have $w=-u\cot{u}$ since $\im{\mathsf{x}\,\ee^{\gamma\,\mathsf{x}}}=0$, which gives
%\begin{align*}
%    \frac{\text{d}\,g(\gamma,\mathsf{x})}{\text{d}\,\mathsf{x}}&=\frac{1}{\mathsf{x}^2u^2((\cot{u}+1)^2+1)}((\gamma\mathsf{x}+u\cot{u})\\&
%    (u^2\cot^2{u}+u^2-u\cot{u})+u^2)>0.\nonumber
%\end{align*}
Here, we used $u^2\cot^2{u}+u^2-u\cot{u}=\frac{u}{\sin{u}}(\frac{u}{\sin{u}}-\cos{u})>0$, and  $\gamma\mathsf{x}+u\cot{u}=\gamma\mathsf{x}-w=\gamma\mathsf{x}-\re{W_0(\mathsf{x}\,\ee^{\gamma\,\mathsf{x}})}=\mathsf{x}(\gamma-g(\gamma,x))>0$, which holds for any $\mathsf{x}\in[\bar{\mathsf{x}},0)$(recall statement~(d)), which finalize our proof for statement~(e).

For proof of statement~(f), notice that statement~(e) explicitly implies that $\max(g(\gamma,\mathsf{x}))=g(\gamma,\mathsf{x}^{\star})$ for any $x\in\real_{<0}$ where $\mathsf{x}^{\star}\ee^{\gamma\,\mathsf{x}^{\star}}=-\frac{1}{\ee}$, which is equivalent to $\mathsf{x}^{\star}=\frac{1}{\gamma}W_0(-\frac{\gamma}{\ee})$ for $\gamma\neq0$, and $\mathsf{x}^{\star}=-\frac{1}{\ee}$ for $\gamma=0$.

%$$------$$
%\solmaz{proof of part (f)}
Proof of statement~(g) is as follows. In statement~(a) we showed that $g(\gamma,\mathsf{x})\rightarrow1$ as $\mathsf{x}\rightarrow0^-$. Moreover, $g(\gamma,\mathsf{x})$ is a continuous ascending function in
$\mathsf{x}\in(-\infty,\frac{1}{\gamma}W_0(-\frac{\gamma}{\ee})]$, and descending function in $\mathsf{x}\in[\frac{1}{\gamma}W_0(-\frac{\gamma}{\ee}),0)$. So, continuity implies that there exists a $\tilde{\mathsf{x}}\in(\bar{\mathsf{x}},\mathsf{x}^{\star})$ such that $g(\gamma,\tilde{\mathsf{x}})=1$, or equivalently $\re{W_0(\tilde{\mathsf{x}}\ee^{\gamma\,\tilde{\mathsf{x}}})}=\tilde{\mathsf{x}}$, and also $g(\gamma,\mathsf{x})>1$ holds for any $\mathsf{x}\in(\tilde{\mathsf{x}},0)$.\boxend %\solmaz{in your statement you have  $\mathsf{x}\in[\tilde{\mathsf{x}},0]}$}
\end{proof}

\begin{proof}[Proof of Theorem~  \ref{thm::delay_effect_scalar}]
%\st{By  Assumption~  \ref{assump::stable_delayfree} we know that for $\tau=0$, system~\eqref{DDE_sys_scalar} is exponentially stable which gives $\alpha+\mathsf{b}<0$.} 
Because by assumption we have $\alpha+\mathsf{b}<0$, $\mathsf{a}>0$ implies that $\mathsf{b}<-\mathsf{a}<0$, resulting in $\gamma>1$ and $\mathsf{x}=\mathsf{a}\tau>0$ for $\tau\in\real_{>0}$. Therefore, invoking Lemma~  \ref{thm::g(k,x)_prop} statement~(b) we get $g(\gamma,\mathsf{x})<\gamma$. Thus,~\eqref{rate_g(x)} implies that system~\eqref{DDE_scalar} is exponentially stable regardless of value of $\tau\in\real_{\geq0}$.
%, which recovers the result of the statement~(a) of Theorem~  \ref{thm::stability_con}. 
Moreover, by taking derivative of $\rho_\tau$ with respect to $\tau$, we obtain
\begin{align}\label{eq::rho_der}
    \frac{\text{d}\,\rho_\tau}{\text{d}\,\tau}=(\frac{\text{d}\,g(\gamma,\mathsf{x})}{\text{d}\,\mathsf{x}})(\frac{\text{d}\,\mathsf{x}}{\text{d}\,\tau})=-\mathsf{a}\,\frac{\text{d}\,g(\gamma,\mathsf{x})}{\text{d}\,\mathsf{x}}.
\end{align}
Lemma~\ref{thm::g(k,x)_prop} part~(c) states that $\frac{\text{d}\,g(\gamma,\mathsf{x})}{\text{d}\,\mathsf{x}}>0$ for any $\gamma>1$ and $\mathsf{x}>0$. Hence, for $\mathsf{a}>0$ we have $\frac{\text{d}\,\rho_\tau}{\text{d}\,\tau}<0$ which concludes our proof of part~(a).

For $\mathsf{a}<0$ and $\mathsf{b}\in\real$, from~\eqref{rate_g(x)} it follows that $\rho_\tau>\rho_0$ if and only if $g(\gamma,\mathsf{a}\tau)>1$. In this case, because of $\mathsf{a}+\mathsf{b}<0$, we have $\gamma<1$ and  $\mathsf{x}=\mathsf{a}\tau<0$ for $\tau\in\real_{>0}$. Therefore, by virtue of statement (g) of Lemma~  \ref{thm::g(k,x)_prop} we have $\rho_\tau>\rho_0$ if and only if $\tau\in[0,\tilde{\tau})\subset[0,\bar{\tau})$ where $\tilde{\tau}$ is the unique solution of $g(\gamma,\mathsf{a}\tau)=1$ in $(0,\bar{\tau})$. 
Additionally, by virtue of part (e) of Lemma~  \ref{thm::g(k,x)_prop}, $\rho_{\tau}$, whose rate of change with respect to $\tau$ is specified by~\eqref{eq::rho_der}, is monotonically increasing (resp. decreasing) with $\tau$ for any $\tau\in[0,\tau^{\star})\subset[0,\bar{\tau})$ (resp. $\tau\in(\tau^{\star},\bar{\tau})\subset[0,\bar{\tau})$) where $\tau^{\star}=\frac{\mathsf{x}^{\star}}{\mathsf{a}}=\frac{1}{\gamma\mathsf{a}}W_0(-\frac{\gamma}{\ee})=-\frac{1}{\mathsf{b}}W_0(\frac{\mathsf{b}}{\mathsf{a}\,\ee})$ for $\mathsf{b}\neq0$ and $\tau^{\star}=\frac{\mathsf{x}^{\star}}{\mathsf{a}}=-\frac{1}{\mathsf{a}\,\ee}$ for  $\mathsf{b}=0$.  Moreover, by virtue of part (f) of  Lemma~\ref{thm::g(k,x)_prop} we conclude that the maximum value of $g(\gamma,\mathsf{x})$ occurs at $\mathsf{x}^{\star}=\mathsf{a}\tau^\star$ where $g(\gamma,\mathsf{x}^{\star})=\frac{-\gamma}{W_0(-\frac{\gamma}{\ee})}$ for $\mathsf{b}\neq0$, which gives 
$\rho_\tau^{\star}=-(1+\frac{1}{W_0(\frac{\mathsf{b}}{\mathsf{a}\ee})})\mathsf{b}$. For $\mathsf{b}=0$, we have $\rho_\tau^{\star}=-\mathsf{a}\,\ee$.\boxend
\end{proof}

\end{appendices}

\end{document}